\documentclass[11pt,letterpaper]{article}
\usepackage{times,mathptmx}
\DeclareMathAlphabet{\mathcal}{OMS}{cmsy}{m}{n}
\usepackage{fullpage}
\usepackage{amsmath,amsfonts,amssymb,amsthm,amscd}
\usepackage{tabularx,booktabs}
\usepackage{hyperref,cite,url}
\hypersetup{pdfpagemode=UseNone,pdfstartview=}
\newcommand{\svs}{\vspace{0.7mm}}
\newcommand{\vs}{\vspace{1.5mm}}
\theoremstyle{plain} 
\newtheorem{theorem}{Theorem}[section]
\newtheorem{lemma}[theorem]{Lemma}

\theoremstyle{definition} 
\newtheorem{definition}{Definition}[section]
\newtheorem{assumption}{Assumption}
\theoremstyle{remark} 

\newcommand{\G}{\mathbb{G}}
\newcommand{\Z}{\mathbb{Z}}
\newcommand{\bits}{\{0,1\}}
\newcommand{\Adv}{\textbf{Adv}}
\newcommand{\mc}[1]{\mathcal{#1}}
\newcommand{\tb}[1]{\textbf{#1}}
\newcommand{\lb}{\linebreak[0]}
\newcommand{\db}{\displaybreak[0]}

\title{Revocable Hierarchical Identity-Based Encryption\\ from Multilinear Maps}

\author{
    Seunghwan Park\footnote{Korea University, Seoul, Korea.
        Email: \texttt{sgusa@korea.ac.kr}.}
    \and
    Dong Hoon Lee\footnote{Korea University, Seoul, Korea.
        Email: \texttt{donghlee@korea.ac.kr}.}
    \and
    Kwangsu Lee\footnote{Sejong University, Seoul, Korea.
        Email: \texttt{kwangsu@sejong.ac.kr}.}
}

\date{}

\begin{document}

\maketitle

\begin{abstract}
In identity-based encryption (IBE) systems, an efficient key delegation
method to manage a large number of users and an efficient key revocation
method to handle the dynamic credentials of users are needed. Revocable
hierarchical IBE (RHIBE) can provide these two methods by organizing the
identities of users as a hierarchy and broadcasting an update key for
non-revoked users per each time period. To provide the key revocation
functionality, previous RHIBE schemes use a tree-based revocation scheme.
However, this approach has an inherent limitation such that the number of
update key elements depends on the number of revoked users.
In this paper, we propose two new RHIBE schemes in multilinear maps that
use the public-key broadcast encryption scheme instead of using the
tree-based revocation scheme to overcome the mentioned limitation. In our
first RHIBE scheme, the number of private key elements and update key
elements is reduced to $O(\ell)$ and $O(\ell)$ respectively where $\ell$ is
the depth of a hierarchical identity. In our second RHIBE scheme, we can
further reduce the number of private key elements from $O(\ell)$ to $O(1)$.
\end{abstract}

\vs \noindent {\bf Keywords:} Hierarchical identity-based encryption,
Key revocation, Key delegation, Multilinear maps.

\newpage

\section{Introduction}

Identity-based encryption (IBE) is a specific type of public-key encryption
(PKE) that uses an identity string of a user (e.g., e-mail address, phone
number) as a public key to simplify the management of public keys
\cite{Shamir84,BonehF01}. IBE can be extended to hierarchical IBE (HIBE) that
supports the delegation of private keys by allowing a parent user to generate
private keys of child users \cite{HorwitzL02,GentryS02}. For the deployment
of IBE (or HIBE) in real environments, an efficient revocation mechanism is
needed to handle dynamically changing credentials (private keys) of users.
Revocable HIBE (RHIBE) is an extension of HIBE that provides both the
delegation of private keys and the revocation of private keys. Although there
already exists a revocable IBE (RIBE) scheme \cite{BoldyrevaGK08}, it is not
easy to directly apply the technique of RIBE to RHIBE since the key
delegation of HIBE makes it hard to handle the revocation.

The first RHIBE scheme was proposed by Seo and Emura \cite{SeoE13e} that uses
a tree-based revocation scheme of Naor, Naor, and Lotspiech \cite{NaorNL01}
for the revocation functionality. To create an update key in this RHIBE
scheme, a user who has a private key should retrieve all update keys of all
ancestors. This method is called as history-preserving updates. After that,
Seo and Emura proposed another RHIBE scheme via history-free updates
\cite{SeoE15}. In this RHIBE scheme via history-free updates, a user can
simply create an update key after retrieving the update key of his parent
only. By using this new approach, they also reduced the size of a private key
from $O(\ell^2 \log N)$ to $O(\ell \log N)$ where $\ell$ is the depth of the
identity and $N$ is the maximum number of users in each level. Recently, Lee
and Park proposed new RHIBE schemes with shorter private keys and update keys
by removing the undesirable multiplicative factor $\ell$ from the size of
private keys and update keys \cite{LeeP16}.

Although the progress of RHIBE is impressive, the size of a private key and
an update key in previous RHIBE schemes still depends on the size of a
private key and a ciphertext in the tree-based revocation scheme. Recently,
Park, Lee, and Lee \cite{ParkLL15} proposed a new RIBE scheme with short keys
from multilinear maps by using the public-key broadcast encryption (PKBE)
scheme of Boneh, Gentry, and Waters \cite{BonehGW05} for the key revocation.
Their new technique enables for RIBE to have a constant number of private key
elements and update key elements. Therefore, we ask the following question in
this paper: \textit{``Can we also reduce the size of keys further in RHIBE by
using the PKBE scheme for the key revocation?"}

\subsection{Our Results}

In this paper, we propose two RHIBE schemes from multilinear
maps\footnote{Note that many candidate multilinear maps are currently broken
\cite{CheonHL+15,CheonFL+16}, but the multilinear map from
indistinguishability obfuscation is still alive \cite{AlbrechtFH+16}.} with
shorter private key elements and shorter update key elements. The followings
are our results:

\vs \noindent \textbf{RHIBE via History-Preserving Updates.} We first
construct an RHIBE scheme via history-preserving updates from three-leveled
multilinear maps by combining the HIBE scheme of Boneh and Boyen (BB-HIBE)
\cite{BonehB04e} and the PKBE scheme of Boneh, Gentry and Waters (BGW-PKBE)
\cite{BonehGW05}. We also prove its security in the selective revocation
model under the multilinear Diffie-Hellman exponent (MDHE) assumption. In
this RHIBE scheme, the number of group elements in a private key, an update
key, and a ciphertext is $O(\ell)$, $O(\ell)$, and $O(\ell)$ respectively
where $\ell$ is the maximum number of hierarchy identity. Note that the
number of private key elements and update key elements in our RHIBE scheme
only depends on the depth of a hierarchy identity.

\vs \noindent \textbf{RHIBE via History-Free Updates.} Next, we present
another RHIBE scheme via history-free updates from three-leveled multilinear
maps with a constant number of private key elements. This RHIBE scheme is
also secure in the selective revocation list model under the MDHE assumption.
In this RHIBE scheme, the number of group elements in a private key, an
update key, and a ciphertext is $O(1)$, $O(\ell)$, and $O(\ell)$
respectively. Compared with our first RHIBE scheme that has $O(\ell)$ group
elements in a private key, our second RHIBE scheme just has $O(1)$ group
elements in a private key. The detailed comparison of RHIBE schemes is given
in Table \ref{tab:rhibe-comp}.

\begin{table*}
\caption{Comparison of revocable hierarchical identity-based encryption schemes}
\label{tab:rhibe-comp}
\vs \small \addtolength{\tabcolsep}{5.2pt}
\renewcommand{\arraystretch}{1.4}
\newcommand{\otoprule}{\midrule[0.09em]}
    \begin{tabularx}{6.50in}{lcccccc}
    \toprule
    Scheme & PP Size & SK Size & UK Size & Model & Maps & Assumption \\
    \otoprule
    SE \cite{SeoE13e}   & $O(\ell)$ & $O(\ell^2 \log N)$ & $O(\ell r \log \frac{N}{r})$
                        & SE-IND & BLM & DBDH \\
    SE (CS) \cite{SeoE15} & $O(\ell)$ & $O(\ell \log N)$ & $O(\ell r \log \frac{N}{r})$
                        & SE-IND & BLM & $q$-Type \\
    SE (SD) \cite{SeoE15} & $O(\ell)$ & $O(\ell \log^2 N)$ & $O(\ell r)$
                        & SRL-IND & BLM & $q$-Type \\
    LP (CS) \cite{LeeP16} & $O(1)$ & $O(\log N)$ & $O(\ell + r \log \frac{N}{r})$
                        & SE-IND & BLM & $q$-Type \\
    LP (SD) \cite{LeeP16} & $O(1)$ & $O(\log^2 N)$ & $O(\ell + r)$
                        & SRL-IND & BLM & $q$-Type \\
    Ours                & $O(N + \lambda \ell)$ & $O(\ell)$ & $O(\ell)$
                        & SRL-IND & MLM & MDHE \\
    Ours                & $O(N + \lambda \ell)$ & $O(1)$ & $O(\ell)$
                        & SRL-IND & MLM & MDHE \\
    \bottomrule
    \multicolumn{7}{p{6.20in}}{
    Let $\lambda$ be a security parameter, $\ell$ be the maximum hierarchical level,
    $N$ be the maximum number of users, and $r$ be the number of revoked users.
    Sizes for public parameters (PP), private keys (SK), and update keys (UK)
    count group elements.
    BLM stands for bilinear maps and MLM stands for multilinear maps.
    }
    \end{tabularx}
\end{table*}

\subsection{Our Techniques}

To construct RHIBE schemes with shorter keys from multilinear maps, we
basically follow the design technique of Park, Lee, and Lee \cite{ParkLL15}
that uses the BGW-PKBE scheme \cite{BonehGW05} instead of the tree-based
revocation system of Naor et al. \cite{NaorNL01}. However, the naive
employment of this technique does not work since the delegation of private
keys should be considered. To solve this problem, we devise new techniques
for RHIBE in multilinear maps. We briefly review the RIBE scheme of Park,
Lee, and Lee \cite{ParkLL15} and then overview our two-level RHIBE scheme for
a simple exposition.

If we simply follow the design strategy of Park et al. \cite{ParkLL15}, a
trusted center which has a master key $\beta_\epsilon, \gamma_\epsilon$
creates a private key for a $1$-level identity $ID_1 = (I_1)$ as %
    $SK_{ID_1} = \big( g_1^{\alpha^{d_1} \gamma_{\epsilon}}
    F_{1,1} (I_1)^{-r_{1,1}}, g_1^{-r_{1,1}} \big)$
and broadcasts a $0$-level update key for time $T$ and a revoked set
$R_{\epsilon}$ as %
    $UK_{T,R_{\epsilon}} = \big( g_1^{\beta_{\epsilon}},
    ( g_1^{\gamma_{\epsilon}} \prod_{j \in SI_{\epsilon}} g_1^{\alpha^{N+1-j}}
    )^{\beta_{\epsilon}} H_1(T)^{r_2}, g_1^{-r_2} \big)$
where $d_i$ is an index assigned to $ID_1$ and $SI_{\epsilon}$ is the set of
receiver indexes. Note that $SK_{ID_1}$ is tied to the private key of PKBE
and $UK_{T,R_\epsilon}$ is tied to the ciphertext header of PKBE. After that,
the $1$-level user of the identity $ID_1$ can delegate his private key to
a $2$-level user with an identity $ID_2 = (I_1, I_2)$ by creating a 2-level
private key as %
    $SK_{ID_2} = \big(
    \{ g_1^{\alpha^{d_1} \gamma_{\epsilon}} F_{1,1} (I_1)^{-r_{1,1}},
       g_1^{-r_{1,1}} \},~
    \{ g_1^{\alpha^{d_2} \gamma_{ID_1}} F_{1,2} (I_2)^{-r_{1,2}}, \lb
       g_1^{-r_{1,2}} \} \big)$.
Next, the $1$-level user broadcasts a $1$-level update key for time $T$ and a
revoked set $R_{ID_1}$ as %
    $UK_{T,R_{ID_1}} = \big(
    \{ ( g_1^{\gamma_{\epsilon}} \prod_{j \in SI_{\epsilon}} g_1^{\alpha^{N+1-j}}
       )^{\beta_{\epsilon}} H_1(T)^{r_{2,1}}, g_1^{-r_{2,1}} \},~ \lb
    \{ ( g_1^{\gamma_{ID_1}} \prod_{j \in SI_{ID_1}} g_1^{\alpha^{N+1-j}}
       )^{\beta_{ID_1}} H_1(T)^{r_{2,2}}, g_1^{-r_{2,2}} \} \big)$.
If $(I_1) \not\in R_{\epsilon}$ and $(I_1,I_2) \not\in R_{ID_1}$, then the
$2$-level user of the identity $ID_2 = (I_1, I_2)$ can derive a decryption key %
    $DK_{ID_2,T} = \big( g_2^{\alpha^{N+1} (\beta_{\epsilon} + \beta_{ID_1})}
    \prod_{i=1}^2 F_{2,i} (I_i)^{r_{1,i}} H_1(T)^{r_2}, g_2^{r_{1,1}},
    g_2^{r_{1,2}}, \lb g_2^{r_2} \big)$
by performing paring operations.

However, there are some problems in the above approach. The first problem is
that the $2$-level user can extract the private key of the $1$-level user
from his private key since $SK_{ID_1}$ is contained in $SK_{ID_2}$. The
second problem is that the master key part $g_1^{ \alpha^{N+1}
(\beta_{\epsilon} + \beta_{ID_1}) }$ of the decryption key $DK_{ID_2,T}$ is
wrongly structured since a random value $\beta_{ID_1}$ that is hidden to a
sender is used. To overcome these problems, we devise a new {\it random
blinding technique} for RHIBE that safely blinds a private key in delegation
and cancels unnecessary random values in decryption. In this technique, the
private key component $g_1^{\alpha^{d} \gamma_\epsilon} F_{1,1}
(I_1)^{-r_{1,1}}$ of $SK_{ID_1}$ is multiplied by a random element
$g_1^{-r_{0,2}}$ and a new element $g_2^{\alpha^{N+1} \beta_{ID_1}}
g_2^{\beta_\epsilon r_{0,2}}$ is included in the private key delegation
process. This newly added element enables to cancel the random values
$r_{0,2}$ and $\beta_{ID_1}$ in the decryption key derivation process. Note
that a $2$-level user who is not revoked in $R_{ID_1}$ only can derive a
correct decryption key which has a master key $g_2^{\alpha^{N+1}
\beta_\epsilon}$ by cancelling the random values. Furthermore, if we encode
the identity of a user carefully, we can reduce the size of private key
elements from $O(\ell^2)$ to $O(\ell)$ where $\ell$ is the hierarchical depth
of the identity. Therefore, we can build an RHIBE scheme via the
history-preserving updates \cite{SeoE13e} in which a private key and an
update key include all private keys and update keys of lower level users from
3-leveled multilinear maps.

To achieve an RHIBE scheme with a constant number of private key elements, we
apply the history-free updates approach of Seo and Emura \cite{SeoE15}. Let
$SK_{ID_1}$ be the $1$-level private key for $ID_1$, $UK_{T,R_\epsilon}$ be
the $0$-level update key, and $DK_{ID_1,T}$ be the $1$-level decryption key
as the same as our first RHIBE scheme. By following this approach, the
$1$-level user with an identity $ID_1 = (I_1)$ simply creates a $2$-level
private key for $ID_2 = (I_1, I_2)$ as %
    $SK_{ID_2} = ( g_1^{\alpha^{d_2} \gamma_{ID_1}} F_{1,2}(I_2)^{-r_{1,2}},
    g_1^{-r_{1,2}} )$.
Next, he creates a $1$-level update key $UK_{T,R_{ID_1}}$ by using his
decryption key $DK_{ID_1,T}$ instead of using the $0$-level update key
$UK_{T,R_\epsilon}$. In this step, we use the random blinding technique to
prevent a collusion attack. That is, the decryption key component
$g_2^{\alpha^{N+1} \beta_\epsilon} F_{2,1}(I_1)^{-r_{1,1}} H_2(T)^{r_0}$ is
safely blinded by a random element $g_2^{-\alpha^{N+1} \beta_{ID_1}}$. Then,
the 1-level update key is formed as %
    $UK_{T,R_{ID_1}} = \big(
    \{ g_2^{\alpha^{N+1} \beta_\epsilon}  g_2^{-\alpha^{N+1} \beta_{ID_1}}
    F_{1,1}(I_1)^{-r_{1,1}} H_2(T)^{r_0}, g_2^{r_1}, g_2^{r_0} \},~
    \{ g_1^{\beta_{ID_1}}, ( g_1^{\gamma_{ID_1}} \prod_{j \in SI_{ID_1}}
    g_1^{\alpha^{N+1-j}} )^{\beta_{ID_1}} \lb H_1(T)^{r_2}, g_1^{-r_2} \} \big)$.
Note that this random blinding element is removed in the decryption key
derivation process. Therefore, we have an RHIBE scheme with shorter private
keys.

\subsection{Related Work}

The concept of IBE was introduced by Shamir to solve the certificate
management problem in PKE \cite{Shamir84}. After the first realization of an
IBE scheme in bilinear maps by Boneh and Franklin \cite{BonehF01}, many IBE
schemes in bilinear maps were proposed \cite{BonehB04e,Waters05,Gentry06,
Waters09}. As mentioned before, providing an efficient revocation mechanism
for IBE is a very important issue. Boneh and Franklin introduced a simple
revocation method for IBE by concatenating an identity $ID$ with time $T$
\cite{BonehF01}. However, this method is not scalable since a trusted center
periodically generates a private key for each user on each time period. The
first scalable RIBE scheme was proposed by Boldyreva, Goyal, and Kumar
\cite{BoldyrevaGK08} by combining the Fuzzy IBE scheme of Sahai and Waters
\cite{SahaiW05} and the complete subtree (CS) scheme of Naor et al.
\cite{NaorNL01}. Many other RIBE schemes also followed this design technique
\cite{LibertV09,SeoE13r}. A different RIBE scheme that uses the subset
difference (SD) scheme instead of using the CS scheme proposed by Lee et al.
\cite{LeeLP14}. Recently, Park, Lee, and Lee proposed a new RIBE scheme from
multilinear maps that has a constant number of private key elements and
update key elements \cite{ParkLL15}.

As mentioned before, the notion of IBE can be extended to HIBE where a
trusted center can delegate the generation of private keys to other users.
After the introduction of HIBE \cite{HorwitzL02}, many HIBE scheme with
different properties were proposed \cite{GentryS02,BonehB04e,BonehBG05,
BoyenW06,GentryH09,Waters09,LewkoW11u,LeePL15}. The first RHIBE scheme was
presented by Seo and Emura \cite{SeoE13e} that combines the BB-HIBE scheme
and the CS scheme. To reduce the size of private keys in the first RHIBE
scheme, Seo and Emura proposed another RHIBE scheme via history-free update
method that combines the BBG-HIBE scheme and the CS (or SD) scheme
\cite{SeoE15}. In previous RHIBE schemes, the maximum hierarchy depth should
be fixed in the setup phase. To overcome this limitation, Ryu et al. proposed
an unbounded RHIBE scheme by using an unbounded HIBE scheme \cite{RyuLPL15}.
Recently, Lee and Park proposed an RHIBE scheme with shorter private keys and
update keys \cite{LeeP16}. To reduce the size of private keys and update
keys, they first presented a new HIBE scheme that supports a short
intermediate private key and build an RHIBE scheme in a modular way.

\section{Preliminaries}

In this section, we review multilinear maps and complexity assumptions in
multilinear maps.

\subsection{Notation}

Let $\lambda$ be a security parameter and $[n]$ be the set $\{1, \ldots, n\}$
for $n \in \Z$. Let $\mc{I}$ be the identity space. A hierarchical identity
$ID$ with a depth $k$ is defined as an identity vector $ID = (I_1, \ldots,
I_k) \in \mc{I}^k$. We let $ID|_j$ be a vector $(I_1, \ldots, I_j)$ of size
$j$ derived from $ID$. If $ID = (I_1, \ldots, I_k)$, then we have $ID =
ID|_k$. We define $ID|_0 = \epsilon$ for simplicity.
A function $\tb{Prefix} (ID|_k)$ returns a set of prefix vectors $\{ ID|_j
\}$ where $1 \leq j \leq k$ where $ID|_k = (I_1, \ldots, I_k) \in \mc{I}^k$
for some $k$. For two hierarchical identities $ID|_i$ and $ID|_j$ with $i <
j$, $ID|_i$ is an ancestor identity of $ID|_j$ and $ID|_j$ is a descendant
identity of $ID|_i$ if $ID|_i \in \tb{Prefix}(ID|_j)$.

\subsection{Leveled Multilinear Maps} \label{sec:ml-maps}

We define generic leveled multilinear maps that are the leveled version of
the cryptographic multilinear maps introduced by Boneh and Silverberg
\cite{BonehS03}. We follow the definition of Garg, Gentry, and Halevi
\cite{GargGH13}.

\begin{definition}[Leveled Multilinear Maps]
We assume the existence of a group generator ${G}$, which takes as input a
security parameter $\lambda$ and a positive integer $k$. Let $\vec{\G} =
(\G_1, \ldots, \G_k)$ be a sequence of groups of large prime order $p >
2^{\lambda}$. In addition, we let $g_i$ be a canonical generator of $\G_i$
respectively. We assume the existence of a set of bilinear maps $\{e_{i,j} :
\G_i \times \G_j \rightarrow \G_{i+j} | i,j \geq 1; i+j \leq k \}$ that have
the following properties:
\begin{itemize}
\item Bilinearity: The map $e_{i,j}$ satisfies the following relation:
    $e_{i,j} (g_i^a, g_j^b) = g_{i+j}^{ab} : \forall a,b \in \Z_p$

\item Non-degeneracy: We have that $e_{i,j} (g_i, g_j) = g_{i+j}$ for each
    valid $i, j$.
\end{itemize}
We say that $\vec{\G}$ is a multilinear group if the group operations in
$\vec{\G}$ as well as all bilinear maps are efficiently computable. We often
omit the subscripts of $e_{i,j}$ and just write $e$.
\end{definition}

\subsection{Complexity Assumptions} \label{sec:comp-assump}

We introduce complexity assumptions in multilinear maps. This assumption is
the multilinear version of the Bilinear Diffie-Hellman Exponent (BDHE)
assumption of Boneh, Gentry, and Waters \cite{BonehGW05}.

\begin{assumption}[Multilinear Diffie-Hellman Exponent, $(k,N)$-MDHE]
Let $(p, \vec{\G}, \{ e_{i,j} | i,j \geq 1; i+j \leq k \})$ be the
description of a $k$-leveled multilinear group of order $p$. Let $g_i$ be a
generator of $\G_i$. The decisional $(k,N)$-MDHE assumption is that if the
challenge tuple
    $$D = \big( g_1, g_1^{a}, g_1^{a^2}, \ldots, g_1^{a^N},
    g_1^{a^{N+2}}, \ldots, g_1^{a^{2N}}, g_1^{c_1}, \ldots, g_1^{c_{k-1}} \big)
    \mbox{ and } Z$$
are given, no PPT algorithm $\mc{A}$ can distinguish $Z = Z_0 = g_k^{ a^{N+1}
\prod_{i=1}^{k-1} c_i}$ from a random element $Z = Z_1 \in \G_k$ with more
than a negligible advantage. The advantage of $\mc{A}$ is defined as
    $\Adv_{\mc{A}}^{(k,N)\text{-}MDHE} (1^\lambda) =
    \big| \Pr[\mc{A}(D,Z_0) = 0] - \Pr[\mc{A}(D,Z_1) = 0] \big|$
where the probability is taken over random choices of $a, c_1, \ldots,
c_{k-1} \in \Z_p$.
\end{assumption}

\begin{assumption}[Three-Leveled Multilinear Diffie-Hellman Exponent,
$(3,N)$-MDHE] Let $(p, \vec{\G}, e_{1,1}, \lb e_{1,2}, e_{2,1})$ be the
description of a three-leveled multilinear group of order $p$. Let $g_i$ be a
generator of $\G_i$. The decisional $(3,N)$-MDHE assumption is that if the
challenge tuple
    $$D = \big( g_1, g_1^{a}, g_1^{a^2}, \ldots, g_1^{a^N},
    g_1^{a^{N+2}}, \ldots, g_1^{a^{2N}}, g_1^b, g_1^c \big)
    \mbox{ and } Z$$
are given, no PPT algorithm $\mc{A}$ can distinguish $Z = Z_0 = g_3^{a^{N+1}
bc}$ from a random element $Z = Z_1 \in \G_3$ with more than a negligible
advantage. The advantage of $\mc{A}$ is defined as
    $\Adv_{\mc{A}}^{(3,N)\text{-}MDHE} (1^\lambda) =
    \big| \Pr[\mc{A}(D,Z_0) = 0] - \Pr[\mc{A}(D,Z_1) = 0] \big|$
where the probability is taken over random choices of $a, b, c \in \Z_p$.
\end{assumption}

\section{Revocable HIBE with History-Preserving Updates} \label{sec:rhibe-hpu}

In this section, we first define the syntax and the security of RHIBE. Next,
we propose an RHIBE scheme with history-preserving updates from three-leveled
multilinear maps and prove its selective security.

\subsection{Definition} \label{sec:rhibe-hpu-def}

Revocable HIBE (RHIBE) is an extension of IBE that provides both the
delegation of private keys and the revocation of private keys \cite{SeoE13e}.
In RHIBE, a user who has a private key $SK_{ID}$ for a hierarchical identity
$ID = (I_1, \ldots, I_{\ell-1})$ can generate a (long-term) private key
$SK_{ID'}$ for a child user with the identity $ID' = (I_1, \ldots, I_\ell)$
where $ID \in \tb{Prefix}(ID')$. The user with $ID$ also periodically
broadcasts an update key $UK_{T,R_{ID}}$ for the set $R_{ID}$ of revoked
child users at a time period $T$. If the child user with $ID'$ who has a
private key $SK_{ID'}$ is not included in the revoked set $R_{ID}$, then he
can derive a (short-term) decryption key $DK_{ID',T}$ from $SK_{ID'}$ and
$UK_{T,R_{ID'}}$. A sender creates a ciphertext $CT_{ID',T}$ that encrypts a
message $M$ for a receiver identity $ID'$ and a time period $T$, and then the
receiver who has a decryption key $DK_{ID',T}$ can obtain the message by
decrypting the ciphertext $CT_{ID',T}$. The syntax of RHIBE is defined as
follows:

\begin{definition}[Revocable HIBE] \label{def:rhibe-hpu-syntax}
A revocable HIBE (RHIBE) scheme that is associated with the identity space
$\mc{I}$, the time space $\mc{T}$, and the message space $\mc{M}$, consists
of seven algorithms \tb{Setup}, \tb{GenKey}, \tb{UpdateKey}, \tb{DeriveKey},
\tb{Encrypt}, \tb{Decrypt}, and \tb{Revoke}, which are defined as follows:
\begin{description}
\item \tb{Setup}($1^\lambda, N, L$): The setup algorithm takes as input a
    security parameter $1^{\lambda}$, the maximum number $N$ of users in
    each depth, and the maximum depth $L$ of the identity. It outputs a
    master key $MK$, a revocation list $RL_{\epsilon}$, a state
    $ST_{\epsilon}$, and public parameters $PP$.

\item \tb{GenKey}($ID|_{\ell}, SK_{ID|_{\ell-1}}, ST_{ID|_{\ell-1}}, PP$):
    The private key generation algorithm takes as input a hierarchical
    identity $ID|_{\ell} = (I_1, \ldots, I_{\ell}) \in \mc{I}^{\ell}$, a
    private key $SK_{ID|_{\ell-1}}$, the state $ST_{ID|_{\ell-1}}$, and
    public parameters $PP$. It outputs a private key $SK_{ID|_{\ell}}$ for
    $ID|_{\ell}$ and updates $ST_{ID|_{\ell-1}}$.

\item \tb{UpdateKey}($T, RL_{ID|_{\ell-1}}, UK_{T,R_{ID|_{\ell-2}}},
    ST_{ID|_{\ell-1}}, PP$): The update key generation algorithm takes as
    input update time $T \in \mc{T}$, a revocation list
    $RL_{ID|_{\ell-1}}$, an update key $UK_{T,R_{ID|_{\ell-2}}}$, the state
    $ST_{ID|_{\ell-1}}$, and the public parameters $PP$. It outputs an
    update key $UK_{T, R_{ID|_{\ell-1}}}$ for $T$ and $R_{ID|_{\ell-1}}$
    where $R_{ID|_{\ell-1}}$ is the set of revoked identities at the time
    $T$.

\item \tb{DeriveKey}($SK_{ID|_{\ell}}, UK_{T,R_{ID|_{\ell-1}}}, PP$): The
    decryption key derivation algorithm takes as input a private key
    $SK_{ID|_{\ell}}$, an update key $UK_{T,R_{ID|_{\ell-1}}}$, and the
    public parameters $PP$. It outputs a decryption key $DK_{ID|_{\ell},T}$
    or $\perp$.

\item \tb{Encrypt}($ID|_{\ell}, T, M, PP$): The encryption algorithm takes
    as input a hierarchical identity $ID|_{\ell} = (I_1, \ldots, I_{\ell})
    \in \mc{I}$, time $T$, a message $M \in \mc{M}$, and the public
    parameters $PP$. It outputs a ciphertext $CT_{ID|_{\ell},T}$ for
    $ID|_{\ell}$ and $T$.

\item \tb{Decrypt}($CT_{ID|_{\ell},T}, DK_{ID'|_{\ell},T'}, PP$): The
    decryption algorithm takes as input a ciphertext $CT_{ID|_{\ell},T}$, a
    decryption key $DK_{ID'|_{\ell},T'}$, and the public parameters $PP$.
    It outputs an encrypted message $M$ or $\perp$.

\item \tb{Revoke}($ID|_{\ell}, T, RL_{ID|_{\ell-1}}, ST_{ID|_{\ell-1}}$):
    The revocation algorithm takes as input a hierarchical identity
    $ID|_{\ell}$ and revocation time $T$, a revocation list
    $RL_{ID|_{\ell-1}}$, and a state $ST_{ID|_{\ell-1}}$. It updates the
    revocation list $RL_{ID|_{\ell-1}}$.
\end{description}
The correctness property of RHIBE is defined as follows: For all $PP$
generated by $\tb{Setup} (1^{\lambda}, N, L)$, $SK_{ID|_{\ell}}$ generated by
$\tb{GenKey} (ID|_{\ell}, SK_{ID|_{\ell-1}}, ST_{ID|_{\ell-1}}, PP)$ for any
$ID|_{\ell}$, $UK_{T,R_{ID|_{\ell-1}}}$ generated by $\tb{UpdateKey} (T,
RL_{ID|_{\ell-1}}, \lb UK_{T, R_{ID|_{\ell-2}}}, ST_{ID|_{\ell-1}}, PP)$ for
any $T$ and $RL_{ID|_{\ell-1}}$, $CT_{ID|_{\ell}, T}$ generated by
$\tb{Encrypt} (ID|_{\ell}, T, \lb M, PP)$ for any $ID|_{\ell}$, $T$, and $M$,
it is required that
\begin{itemize}
\item If $(ID|_{\ell} \notin R_{ID|_{\ell-1}})$, then $\tb{DeriveKey}
    (SK_{ID|_{\ell}}, UK_{T, R_{ID|_{\ell-1}}}, PP) = DK_{ID|_{\ell},T}$.

\item If $(ID|_{\ell} \in R_{ID|_{\ell-1}})$, then $\tb{DeriveKey}
    (SK_{ID|_{\ell}}, UK_{T, R_{ID|_{\ell-1}}}, PP) = \perp$ with all but
    negligible probability.

\item If $(ID|_{\ell} = ID'|_{\ell}) \wedge (T = T')$, then $\tb{Decrypt}
    (CT_{ID|_{\ell},T}, DK_{ID'|_{\ell},T'}, PP) = M$.

\item If $(ID|_{\ell} \neq ID'|_{\ell}) \vee (T \neq T')$, then
    $\tb{Decrypt} (CT_{ID|_{\ell},T}, DK_{ID'|_{\ell},T'}, PP) = \perp$
    with all but negligible probability.
\end{itemize}
\end{definition}

The security model of RHIBE with history-preserving updates was defined by
Seo and Emura \cite{SeoE13e}. For the security proof our RHIBE scheme, we
define a selective revocation list model where an adversary initially submits
the set of revoked identities. This weaker model was introduced by Boldyreva,
et al. \cite{BoldyrevaGK08} and used in other schemes \cite{LeeLP14,ParkLL15,
SeoE15}. In this paper, we follow the selective revocation list model refined
by Seo and Emura \cite{SeoE15}.
In this security model of RHIBE, an adversary initially submits a challenge
identity $ID^*|_{\ell^*}$, challenge time $T^*$, and a revoked identity set
$R^*$ at the time $T^*$. After receiving public parameters $PP$, the
adversary can adaptively request private keys, update keys, decryption keys,
and revocations with some restrictions. In the challenge step, the adversary
submits challenge messages $M_0^*, M_1^*$ and the challenger creates a
challenge ciphertext $CT^*$ that encrypts one of the challenge messages. The
adversary continually requests private keys, update keys, and decryption
keys. Finally, if the adversary correctly guesses the encrypted message, then
he wins the game. The details of the security model is described as follows:

\begin{definition}[Selective Revocation List Security, SRL-IND]
\label{def:rhibe-hpu-srlind} %
The SRL-IND security of RHIBE is defined in terms of the following experiment
between a challenger $\mc{C}$ and a PPT adversary $\mc{A}$:
\begin{enumerate}
\item \tb{Init}: $\mc{A}$ initially submits a challenge identity
    $ID^*|_{\ell^*} \in \mc{I}^{\ell^*}$, challenge time $T^* \in \mc{T}$,
    and a revoked identity set $R^* \subseteq \mc{I}^{\ell}$ at the time
    $T^*$.

\item \tb{Setup}: $\mc{C}$ generates a master key $MK$, a revocation list
    $RL_{\epsilon}$, a state $ST_{\epsilon}$, and public parameters $PP$ by
    running $\tb{Setup} (1^{\lambda}, N, L)$. It keeps $MK, RL_{\epsilon},
    ST_{\epsilon}$ to itself and gives $PP$ to $\mc{A}$.

\item \tb{Phase 1}: $\mc{A}$ adaptively requests a polynomial number of
    queries. These queries are processed as follows:
    \begin{itemize}
    \item \tb{Private key}: If this is a private key query for a
        hierarchical identity $ID|_\ell$, then it gives the private key
        $SK_{ID|_\ell}$ to $\mc{A}$ by running $\tb{GenKey} (ID|_{\ell},
        SK_{ID|_{\ell-1}}, ST_{ID|_{\ell-1}}, PP)$ with the restriction:
        If $ID|_{\ell}$ is a prefix of $ID^*|_{\ell^*}$ where $\ell \leq
        \ell^*$, then the revocation query for $ID|_{\ell}$ or one of its
        ancestors must be queried at some time $T$ where $T \leq T^*$.

    \item \tb{Update key}: If this is an update key query for a
        hierarchical identity $ID|_{\ell-1}$ and time $T$, then it gives
        the update key $UK_{T,R_{ID|_{\ell-1}}}$ to $\mc{A}$ by running
        $\tb{UpdateKey} (T, RL_{ID|_{\ell-1}}, UK_{T,R_{ID|_{\ell-2}}},
        ST_{ID|_{\ell-1}}, PP)$ with the restriction: If $T = T^*$, then
        the revoked identity set of $RL_{ID|_{\ell-1}}$ at the time $T^*$
        should be equal to a revoked identity set derived from $R^*$.

    \item \tb{Decryption key}: If this is a decryption key query for a
        hierarchical identity $ID|_{\ell}$ and time $T$, then it gives
        the decryption key $DK_{ID|_\ell,T}$ to $\mc{A}$ by running
        $\tb{DeriveKey} (SK_{ID|_{\ell}}, UK_{T,R_{ID|_{\ell-1}}}, PP)$
        with the restriction: The decryption key query for
        $ID^*|_{\ell^*}$ and $T^*$ cannot be queried.

    \item \tb{Revocation}: If this is a revocation query for a
        hierarchical identity $ID|_{\ell}$ and revocation time $T$, then
        it updates the revocation list $RL_{ID|_{\ell-1}}$ by running
        $\tb{Revoke} (ID|_{\ell}, T, RL_{ID|_{\ell-1}},
        ST_{ID|_{\ell-1}})$ with the restriction: The revocation query
        for time $T$ cannot be queried if the update key query for the
        time $T$ was already requested.
    \end{itemize}
Note that $\mc{A}$ is allowed to request the update key query and the
revocation query in non-decreasing order of time, and an update key $UK_{T,
R_{ID|_{\ell-1}}}$ implicitly includes a revoked identity set
$R_{ID|_{\ell-1}}$ derived from $RL_{ID|_{\ell-1}}$.

\item \tb{Challenge}: $\mc{A}$ submits two challenge messages $M_0^*, M_1^*
    \in \mc{M}$ with equal length. $\mc{C}$ flips a random coin $b \in
    \bits$ and gives the challenge ciphertext $CT^*$ to $\mc{A}$ by running
    $\tb{Encrypt} (ID^*|_{\ell^*}, T^*, M_b^*, PP)$.

\item \tb{Phase 2}: $\mc{A}$ may continue to request a polynomial number of
    private keys, update keys, and decryption keys subject to the same
    restrictions as before.

\item \tb{Guess}: Finally, $\mc{A}$ outputs a guess $b' \in \bits$, and
    wins the game if $b = b'$.
\end{enumerate}
The advantage of $\mc{A}$ is defined as $\Adv_{RHIBE, \mc{A}
}^{SRL\text{-}IND} (1^\lambda) = \big| \Pr [b = b'] - \frac{1}{2} \big|$
where the probability is taken over all the randomness of the experiment. An
RHIBE scheme is SRL-IND secure if for all PPT adversary $\mc{A}$, the
advantage of $\mc{A}$ in the above experiment is negligible in the security
parameter $\lambda$.
\end{definition}

\subsection{Building Blocks}

Let $\mc{I} = \bits^{l_1}$ be the identity space and $\mc{T} = \bits^{l_2}$
be the time space where $l_1 = 2\lambda$ and $l_2 = \lambda$ for a security
parameter $\lambda$. Let $ID = (I_1, \ldots, I_k)$ be an hierarchical
identity. We define $\tb{EncodeCID}(ID)$ as a function that takes as input
$ID = (I_1, \ldots, I_k)$ and outputs a concatenated identity $CID = (CI_1,
\ldots, CI_k)$ where $CI_j = H(I_1 \| \cdots \| I_j)$, $\|$ denotes the
concatenation of two strings, and $H$ is a collision-resistant hash function.
This encoding function has an interesting property such that if $ID \not\in
\tb{Prefix}(ID')$ then $CI_k \neq CI'_i$ for all $i \in [k]$ except with
negligible probability \cite{LeeP16}.

We describe a modified HIBE scheme of Boneh and Boyen \cite{BonehB04e} that
additionally takes a time period in multilinear groups. Note that we define
the key-encapsulation mechanism (KEM) version of HIBE.

\begin{description}
\item [\tb{HIBE.Setup}($GDS_{MLM}, L$):] Let $GDS_{MLM} = (p, \vec{\G} =
    (\G_1, \G_2, \G_3), \{ e_{i,j} \}, g_1, g_2, g_3)$ be the description
    of a multilinear group and $L$ be the maximum depth of the hierarchical
    identity.
    \begin{enumerate}
    \item It first selects random elements $\{ f_{1,i,0} \}_{1 \leq i
        \leq L}, \{ f_{1,i,j,b} \}_{1 \leq i \leq L, 1 \leq j \leq l_1, b
        \in \bits} \in \G_1$. It also selects random $h_{1,0}, \{
        h_{1,j,b} \}_{1 \leq j \leq l_2, b \in \bits} \in \G_1$.
        Let $\vec{f}_{k,i} = \big( f_{k,i,0}, \{ f_{k,i,j,b} \} \big)$
        and $\vec{h}_k = \big( h_{k,0}, \{ h_{k,j,b} \} \big)$ for a
        multi-linear level $k$. Note that $\{ \vec{f}_{2,i} \}_{1 \leq i
        \leq L}$ and $\vec{h}_2$ can be obtained from $\{ \vec{f}_{1,i}
        \}_{1 \leq i \leq L}$ and $\vec{h}_1$ by performing pairing
        operations.

    \item Next, it defines $F_{k,i}(CI) = f_{k,i,0} \prod_{j=1}^{l_1}
        f_{k,i,j,CI[j]}$ and $H_k(T) = h_{k,0} \prod_{j=1}^{l_2}
        h_{k,i,j,T[j]}$ where $CI[j]$ is a bit value at the position $j$
        and $T[j]$ is a bit value at the position $j$.

    \item It selects a random exponent $\alpha \in \Z_p$ and outputs a
        master key $MK = \alpha$ and public parameters
        \begin{align*}
        PP = \Big(
            GDS_{MLM},~ \{ \vec{f}_{1,i} \}_{1 \leq i \leq L},~ \vec{h}_1,~
            \Lambda = g_3^{\alpha}
        \Big).
        \end{align*}
    \end{enumerate}

\item [\tb{HIBE.GenKey}($ID|_{\ell}, T, MK, PP$):] Let $ID|_{\ell} = (I_1,
    \ldots, I_{\ell}) \in \mc{I}^{\ell}$, $T \in \mc{T}$, and $MK =
    \alpha$. It obtains $CID|_\ell = (CI_1, \ldots, CI_\ell)$ by calling
    $\tb{EncodeCID}(ID|_\ell)$. It selects random exponents $r_1, \ldots,
    r_{\ell}, r_{L+1} \in \Z_p$ and outputs a private key
    \begin{align*}
    SK_{ID|_{\ell},T} = \Big(
        D_0 = g_2^{\alpha} \prod_{i=1}^\ell
              F_{2,i}(CI_i)^{r_i} \cdot H_2(T)^{r_{L+1}},~
        \big\{ D_i = g_2^{-r_i} \big\}_{i=1}^\ell,~
        D_{L+1} = g_2^{-r_{L+1}}
    \Big) \in \G_2^{\ell+2}.
    \end{align*}

\item [\tb{HIBE.RandKey}($SK_{ID|_{\ell},T}, PP$):] Let $SK_{ID|_{\ell},T}
    = (D'_0, \ldots, D'_\ell, D'_{L+1})$ and $ID = (I_1, \ldots,
    I_{\ell})$. It selects random exponents $r_1, \ldots, r_\ell, r_{L+1}
    \in \Z_p$ and outputs a randomized private key %
    $SK_{ID|_{\ell},T} = \big( D_0 = D'_0 \cdot \prod_{i=1}^\ell
        F_{2,i}(CI_1)^{r_i} \cdot H_2(T)^{r_{L+1}},~ \big\{ D_i = D'_i
        \cdot g_2^{-r_i}\big\}_{i=1}^\ell,~ D_{L+1} = D'_{L+1} \cdot
        g_2^{-r_{L+1}} \big)$.

\item [\tb{HIBE.Delegate}($ID|_{\ell}, SK_{ID|_{\ell-1},T}, PP$):] Let
    $SK_{ID|_{\ell-1}} = (D'_0, \ldots, D'_{\ell})$ and $ID|_\ell = (I_1,
    \ldots, I_{\ell})$. It obtains $CID|_\ell = (CI_1, \ldots, CI_\ell)$ by
    calling $\tb{EncodeCID}(ID|_\ell)$. It selects random a exponent
    $r_{\ell} \in \Z_p$ and creates a temporal private key %
    $TSK_{ID|_{\ell},T} = \big( D_0 = D'_0 \cdot F_{2,\ell}
        (CI_{\ell})^{r_\ell},~ \big\{ D_i = D'_i \big\}_{i=1}^{\ell-1},~
        D_{\ell} = g_2^{-r_{\ell}},~ D_{L+1} = D'_{L+1} \big)$. %
    Next, it outputs a delegated private key $SK_{ID|_{\ell},T}$ by running
    $\tb{HIBE.RandKey} (TSK_{ID|_{\ell},T}, PP)$.

\item [\tb{HIBE.Encrypt}($ID|_{\ell}, T, s, PP$):] Let $ID|_{\ell} = (I_1,
    \ldots, I_{\ell})$ and $s$ is a random exponent in $\Z_p$. It obtains
    $CID|_\ell = (CI_1, \ldots, CI_\ell)$ by calling
    $\tb{EncodeCID}(ID|_\ell)$. It outputs a ciphertext header %
    \begin{align*}
    CH_{ID|_\ell,T} = \Big(
        C_0 = g_1^s,~
        \big\{ C_i = F_{1,i}(CI_i)^s \big\}_{i=1}^{\ell},~
        C_{L+1} = H_1(T)^s
    \Big) \in \G_1^{\ell+2}
    \end{align*}
    and a session key $EK = \Lambda^s$.

\item [\tb{HIBE.Decrypt}($CH_{ID|_{\ell}, T}, SK_{ID'|_{\ell}, T'}, PP$):]
    Let $CH_{ID|_\ell,T} = (C_0, \{ C_i \}_{i=1}^\ell, C_{L+1})$ and
    $SK_{ID'|_\ell,T'} = (D_0, \{ D_i \}_{i=1}^\ell, D_{L+1})$. If
    $(ID|_\ell = ID'|_\ell) \wedge (T = T')$, then it outputs the session
    key $EK$ by computing $e_{1,2} (C_0, D_0) \cdot \prod_{i=1}^{\ell}
    e_{1,2} (C_i, D_i) \cdot e_{1,2} (C_{L+1}, D_{L+1})$. Otherwise, it
    outputs $\perp$.
\end{description}

Let $\mc{N} = \{ 1, \ldots, N \}$ where $N$ is the (polynomial) number of
users. We describe the PKBE scheme of Boneh, Gentry, and Waters
\cite{BonehGW05}.

\begin{description}
\item [\tb{PKBE.Setup}($GDS_{BLM}, N$):] Let $GDS_{BLM} = ((p, \vec{\G} =
    (\G_1, \G_2), e_{1,1}), g_1, g_2)$ and $N$ be the maximum number of
    users. It selects random exponents $\alpha, \gamma \in \Z_p$ and
    outputs a master key $MK = (\alpha, \gamma)$, an element $Y =
    g_1^{\gamma}$,
    and public parameters %
    \begin{align*}
    PP = \Big( GDS_{BLM},~
        \big\{ X_j = g_1^{\alpha^j} \big\}_{1 \leq j, j \neq N+1 \leq 2N},~
        \Gamma = g_2^{\alpha^{N+1}}
    \Big).
    \end{align*}

\item [\tb{PKBE.GenKey}($d, MK, PP$):] Let $d \in \mc{N}$ be an index and
    $MK = (\alpha, \gamma)$. It outputs a private key %
    $SK_d = \big( K = g_1^{\alpha^{d} \gamma} \big).$

\item [\tb{PKBE.Encrypt}($S, \beta, Y, PP$):] Let $S$ be a set of receiver
    indexes, $\beta$ be a random exponent in $\Z_p$, and $Y = g_1^{\gamma}$
    be a group element in $\G_1$. It outputs a ciphertext header
    \begin{align*}
    CH_S = \Big(
        E_0 = g_1^{\beta},~
        E_1 = \big( Y \prod_{j \in S} X_{N+1-j} \big)^{\beta}
    \Big)
    \end{align*}
    and a session key $EK = \Gamma^{\beta}$.

\item [\tb{PKBE.Decrypt}($CH_S, SK_d, PP$):] Let $CH_S = (E_0, E_1)$ and
    $SK_d = K$. If $d \in S$, then it outputs the session key $EK$ by
    computing $e(X_d, E_1) \cdot e(E_0, K \cdot \prod_{j \in S, j \neq d}
    X_{N+1-j+d} )^{-1}$. Otherwise, it outputs $\perp$.
\end{description}

\begin{theorem}[\cite{BonehGW05}] \label{thm:bgw-pkbe-indcpa}
The BGW-PKBE scheme is selectively secure under chosen plaintext attacks if
the $N$-BDHE assumption holds.
\end{theorem}

\subsection{Construction} \label{sec:rhibe-hpu-scheme}

By using the HIBE and PKBE schemes that are described in the previous
section, we can build an RHIBE scheme. Our RHIBE scheme with
history-preserving updates in multilinear maps is described as follows:

\begin{description}
\item [\tb{RHIBE.Setup}($1^\lambda, N, L$):] Let $N$ be the maximum number
    users in each depth and $L$ be the maximum depth of the hierarchical
    identity.
    \begin{enumerate}
    \item It first generates a multilinear group $\vec{\G} = (\G_1, \G_2,
        \G_3)$ of prime order $p$. Let $GDS_{MLM} = (p, \vec{\G}, \{
        e_{1,1}, e_{1,2}, e_{2,1} \}, g_1, g_2, g_3)$ be the description
        of the multilinear group where $g_1, g_2, g_3$ are generators of
        $\G_1, \G_2, \G_3$ respectively.

    \item It obtains $MK_{HIBE}, PP_{HIBE}$ by running $\tb{HIBE.Setup}
        (GDS_{MLM}, L)$. It also obtains $MK_{BE} = (\alpha, \gamma),
        PP_{BE}$ by running $\tb{PKBE.Setup}(GDS_{MLM}, N)$.

    \item It selects a random exponent $\beta_{\epsilon} \in \Z_p$ and
        saves $(\beta_{\epsilon}, \gamma_{\epsilon})$ to $ST_{\epsilon}$
        where $\beta_{\epsilon} = \beta_{ID_0}$ and $\gamma_{\epsilon} =
        \gamma_{ID_0} = \gamma$. It outputs a master key $MK = \alpha$,
        an empty revocation list $RL_{\epsilon}$, a state
        $ST_{\epsilon}$, and public parameters
        \begin{align*}
        PP &= \Big(
           GDS_{MLM},~ PP_{HIBE},~ PP_{BE},~
           g_2^{\alpha^{N+1}},~ g_2^{\beta_{\epsilon}},~
           \Omega = g_3^{\alpha^{N+1} \beta_{\epsilon}}
        \Big).
        \end{align*}
    \end{enumerate}

\item [\tb{RHIBE.GenKey}($ID|_{\ell}, SK_{ID|_{\ell-1}}, ST_{ID|_{\ell-1}},
    PP$):] Let $ID|_{\ell} = (I_1,\ldots, I_{\ell}) \in \mc{I}^{\ell}$ and
    $SK_{ID|_{\ell-1}} = \big( \{ d_i, LSK'_i \}_{i=1}^{\ell-1} \big)$
    where $LSK'_i = ( K'_{i,0}, K'_{i,1}, R'_{i,0}, R'_{i,1} )$ and $\ell
    \geq 1$. It obtains $CID|_\ell = (CI_1, \ldots, CI_\ell)$ by calling
    $\tb{EncodeCID}(ID|_\ell)$. Recall that $SK_{ID|_0}$ is empty.
    \begin{enumerate}
    \item If a tuple $(\beta_{ID|_{\ell-1}}, \gamma_{ID|_{\ell-1}})$
        exist in $ST_{ID|_{\ell-1}}$, then it retrieves
        $(\beta_{ID|_{\ell-1}}, \gamma_{ID|_{\ell-1}})$ from
        $ST_{ID|_{\ell-1}}$. Otherwise, it selects random exponents
        $\beta_{ID|_{\ell-1}}, \gamma_{ID|_{\ell-1}} \in \Z_p$ and saves
        $(\beta_{ID|_{\ell-1}}, \gamma_{ID|_{\ell-1}})$ to
        $ST_{ID|_{\ell-1}}$.

    \item If $\ell \geq 2$, then it selects a random exponent
        $r_{\ell-1,2} \in \Z_p$ and creates an updated level private key
        \begin{align*}
        LSK_{\ell-1} = \Big(
        &   K_{\ell-1,0} = K'_{\ell-1,0} \cdot g_1^{-r_{\ell-1,2}},~
            K_{\ell-1,1} = K'_{\ell-1,1},~ \\
        &   R_{\ell-1,0} = g_2^{ \beta_{ID|_{\ell-1}} },~
            R_{\ell-1,1} = \big( g_2^{\alpha^{N+1}} \big)^{\beta_{ID|_{\ell-1}}}
                     \cdot (g_2^{\beta_{ID|_{\ell-2}}})^{r_{\ell-1,2}}
        \Big)
        \end{align*}
        where $g_2^{\alpha^{N+1}}$ can be retrieved from $PP$ and
        $g_2^{\beta_{ID|_{\ell-2}}}$ can be retrieved from $LSK_{\ell-2}$
        or $PP$.

    \item It assigns a unique index $d_{\ell} \in \mc{N}$ to the identity
        $ID|_{\ell}$ and adds a tuple $(ID|_\ell, d_{\ell})$ to $ST_{ID|_
        {\ell-1}}$.
        It obtains $SK_{BE,d_{\ell}} = K_{BE} = g_1^{\alpha^{d_\ell}
        \gamma_{ID|_{\ell-1}}}$ by running $\tb{PKBE.GenKey} (d_{\ell},
        \gamma_{ID|_{\ell-1}}, PP_{BE})$. Next, it selects a random
        exponent $r_{\ell,1} \in \Z_p$ and creates a level private key
        \begin{align*}
        LSK_{\ell} = \Big(
            K_{\ell,0} = K_{BE} \cdot F_{1,\ell}(CI_{\ell})^{-r_{\ell,1}},~
            K_{\ell,1} = g_1^{-r_{\ell,1}},~
            R_{\ell,0} = 1_{\G_2},~
            R_{\ell,1} = 1_{\G_2}
        \Big) \in \G_1^2 \times \G_2^2.
        \end{align*}

    \item Finally, it outputs a private key $SK_{ID|_\ell} = \big( \big\{
        d_i, LSK_i = LSK'_i \big\}_{i=1}^{\ell-2}, \big\{ d_i, LSK_i
        \big\}_{i=\ell-1}^{\ell} \big)$.
    \end{enumerate}

\item [\tb{RHIBE.UpdateKey}($T, RL_{ID|_{\ell-1}}, UK_{T,
    R_{ID|_{\ell-2}}}, ST_{ID|_{\ell-1}}, PP$):] Let $UK_{T,
    R_{ID|_{\ell-2}}} = \big( \{ SI_{ID|_i}, LUK'_i \}_{i=0}^{\ell-2}
    \big)$ where $LUK_i = ( U_{i,0}, U_{i,1}, U_{i,2} )$ and $\ell \geq 1$.
    Recall that $UK_{T, ID|_{-1}}$ is empty.
    \begin{enumerate}
    \item It defines a revoked set $R_{ID|_{\ell-1}}$ of user identities
        at time $T$ from $RL_{ID|_{\ell-1}}$. From $R_{ID|_{\ell-1}}$, it
        defines a revoked index set $RI_{ID|_{\ell-1}} \subseteq \mc{N}$
        by using $ST_{ID|_{\ell-1}}$ since $ST_{ID|_{\ell-1}}$ contains
        $(ID|_\ell, d_{\ell})$. After that, it defines a non-revoked
        index set $SI_{ID|_{\ell-1}} = \mc{N} \setminus
        RI_{ID|_{\ell-1}}$.

    \item It retrieves $(\beta_{ID|_{\ell-1}}, \gamma_{ID|_{\ell-1}})$
        from $ST_{ID|_{\ell-1}}$. It obtains $CH_{BE} = (E_0, E_1)$ by
        running $\tb{PKBE.Encrypt} \lb (SI_{ID|_{\ell-1}},
        \beta_{ID|_{\ell-1}}, Y_{ID|_{\ell-1}} =
        g_1^{\gamma_{ID|_{\ell-1}}}, PP_{BE})$. It selects a random
        exponent $r_{\ell-1} \in \Z_p$ and creates a level update key
        \begin{align*}
        LUK_{\ell-1} = \Big(
            U_{\ell-1,0} = E_0,~
            U_{\ell-1,1} = E_1 \cdot H_1(T)^{r_{\ell-1}},~
            U_{\ell-1,2} = g_1^{-r_{\ell-1}}
        \Big) \in \G_1^3.
        \end{align*}

    \item Finally, it outputs an update key $UK_{T,R_{ID|_{\ell-1}}} =
        \big( \big\{ SI_{ID|_i},~ LUK_i = LUK'_i \big\}_{i=0}^{\ell-2},
        \big\{ SI_{ID|_{\ell-1}},~ LUK_{\ell-1} \big\}  \big)$.
    \end{enumerate}

\item [\tb{RHIBE.DeriveKey}($SK_{ID|_\ell}, UK_{T,R_{ID|_{\ell-1}}}, PP$):]
    Let $SK_{ID|_\ell} = ( \{ d_i, LSK_i \}_{i=1}^{\ell} )$ where $LSK_i =
    ( K_{i,0}, K_{i,1}, R_{i,0}, R_{i,1} )$ and $\ell \geq 1$, and
    $UK_{T,R_{ID|_{\ell-1}}} = ( \{ SI_{ID|_i}, LUK_i \}_{i=0}^{\ell-1} )$
    where $LUK_i = ( U_{i,0}, U_{i,1}, U_{i,2} )$.
    If $ID|_{\ell} \in R_{ID|_{\ell-1}}$, then it outputs $\perp$ since the
    identity $ID|_\ell$ is revoked. Otherwise, it proceeds as follows:
    \begin{enumerate}
    \item For each $i \in [\ell]$, it retrieves $\{ d_i, LSK_i =
        (K_{i,0}, K_{i,1}, R_{i,0}, R_{i,1}) \}$ and $\{ SI_{ID|_{i-1}},
        LUK_{i-1} = ( U_{i-1,0}, U_{i-1,1}, \lb U_{i-1,2} ) \}$ and
        computes the following components
        \begin{align*}
        A_{i,0} &= e_{1,1} (X_{d_i}, U_{i-1,1}) \cdot
                   e_{1,1} \big( U_{i-1,0}, K_{i,0}
                   \prod_{j \in SI_{ID|_{i-1}}, j \neq d_i} X_{N+1 -j + d_i}
                   \big)^{-1},~ \\
        A_{i,1} &= e_{1,1} (U_{i-1,0}, K_{i,1}),~
        A_{i,2}  = e_{1,1} (X_{d_i}, U_{i-1,2}).
        \end{align*}

    \item Next, it derives a temporal decryption key
        \begin{align*}
        TDK_{ID|_{\ell}, T} = \Big(
            D_0 = \prod_{i=1}^{\ell} A_{i,0} \cdot
                  \prod_{i=1}^{\ell-1} R_{i,1}^{-1},~
            \big\{ D_i = A_{i,1} \big\}_{i=1}^{\ell},~
            D_{L+1} = \prod_{i=1}^{\ell} A_{i,2}
        \Big) \in \G_2^{\ell+2}.
        \end{align*}

    \item Finally, it outputs a decryption key $DK_{ID|_{\ell}, T}$ by
        running $\tb{HIBE.RandKey} (TDK_{ID|_{\ell},T}, PP_{HIBE})$.
    \end{enumerate}

\item [\tb{RHIBE.Encrypt}($ID|_{\ell}, T, M, PP$):] Let $ID|_{\ell} = (I_1,
    \ldots, I_{\ell})$. It first chooses a random exponent $s \in \Z_p$ and
    obtains $CH_{HIBE}$ by running $\tb{HIBE.Encrypt} (ID|_{\ell}, T, s,
    PP_{HIBE})$. It outputs a ciphertext $CT_{ID|_\ell,T} = \big( C =
    \Omega^s \cdot M,~ CH_{HIBE} \big)$.

\item [\tb{RHIBE.Decrypt}($CT_{ID|_\ell,T}, DK_{ID'|_\ell,T'}, PP$):] Let
    $CT_{ID|_\ell,T} = (C, CH_{HIBE})$. If $(ID|_\ell = ID'|_\ell) \wedge
    (T = T')$, then it obtains $EK_{HIBE}$ by running $\tb{HIBE.Decrypt}
    (CH_{HIBE}, DK_{ID',T'}, PP_{HIBE})$ and outputs the message $M$ by
    computing $M = C \cdot EK^{-1}_{HIBE}$. Otherwise, it outputs $\perp$.

\item [\tb{RHIBE.Revoke}($ID|_{\ell}, T, RL_{ID|_{\ell-1}},
    ST_{ID|_{\ell-1}}$):] If $(ID|_\ell, -) \notin ST_{ID|_{\ell-1}}$, then
    it outputs $\perp$ since the private key of $ID|_{\ell}$ was not
    generated. Otherwise, it updates $RL_{ID|_{\ell-1}}$ by adding
    $(ID|_{\ell}, T)$ to $RL_{ID|_{\ell-1}}$.
\end{description}

\subsection{Correctness}

Let $SK_{ID|_{\ell}} = ( \{ d_i, LSK_i \}_{i=1}^\ell )$ be a private key for
$ID|_{\ell} = (I_1, \ldots, I_\ell)$ where $LSK_i = (K_{i,0}, K_{i,1},
R_{i,0}, R_{i,1})$ and $d_i$ is an index for $ID|_i = (I_1, \ldots, I_i)$.
Let $UK_{T,R_{ID|_{\ell-1}}} = ( \{ SI_{ID|_i}, LUK_i \}_{i=0}^{\ell-1} )$ be
an update key for time $T$ and a revoked set $R_{ID|_{\ell-1}}$. If $ID|_\ell
\not\in R_{ID|_{\ell-1}}$, then we obtain the following equations
    \begin{align*}
    A_{i,0}
    &=  e_{1,1} (X_{d_i}, U_{i-1,1}) \cdot
        e_{1,1} \Big( U_{i-1,0}, K_{i,0} \cdot
            \prod_{j \in SI_{ID|_{i-1}}, j \neq d_i} X_{N+1-j+d_i} \Big)^{-1} \\
    &=  EK_{BE} \cdot e_{1,1} (X_{d_i}, H_1(T)^{r_i}) \cdot
        e_{1,1} (E_0, F_{1,i}(CI_i)^{r_{i,1}} g_1^{r_{i,2}}) \\
    &=  g_2^{\alpha^{N+1} \beta_{ID|_{i-1}}} H_2(T)^{\alpha^{d_i} r_i}
        F_{2,i}(CI_i)^{\beta_{ID|_{i-1}} r_{i,1}}
        g_2^{\beta_{ID|_{i-1}} r_{i,2}},~ \db \\
    A_{i,1}
    &=  e_{1,1} (U_{i-1,0}, K_{i,1})
     =  e_{1,1} \big( g_1^{\beta_{ID|_{i-1}}}, g_1^{-r_{i,1}} \big)
     =  g_2^{-\beta_{ID|_{i-1}} r_{i,1}},~ \\
    A_{i,2}
    &=  e_{1,1} (X_{d_i}, U_{i-1,2})
     =  e_{1,1} \big( g_1^{\alpha^{d_i}}, g_1^{-r_i} \big)
     =  g_2^{-\alpha^{d_i} r_i}
    \end{align*}
from the correctness of the PKBE scheme. The decryption key derivation
algorithm correctly derives a temporal decryption key as
    \begin{align*}
    D_0
    &=  \prod_{i=1}^{\ell} A_{i,0} \cdot \prod_{i=1}^{\ell-1} R_{i,1}^{-1} \\
    &=  \prod_{i=1}^{\ell}
        g_2^{\alpha^{N+1} \beta_{ID|_{i-1}}}
        H_2(T)^{\alpha^{d_i} r_i} F_{2,i}(CI_i)^{\beta_{ID|_{i-1}} r_{i,1}}
        g_2^{\beta_{ID|_{i-1}} r_{i,2}} \cdot
        \prod_{i=1}^{\ell-1}
        g_2^{-\alpha^{N+1} \beta_{ID|_i}} g_2^{-\beta_{ID|_{i-1}} r_{i,2}} \\
    &=  g_2^{\alpha^{N+1} \beta_{\epsilon}} \cdot
        \prod_{i=1}^{\ell} F_{2,i}(CI_i)^{\beta_{ID|_{i-1}} r_{i,1}} \cdot
        H_2(T)^{\sum_{i=1}^{\ell} \alpha^{d_i} r_i},~ \db \\
    D_{L+1}
    &=  \prod_{i=1}^{\ell} A_{i,2}
     =  \prod_{i=1}^{\ell} g_2^{-\alpha^{d_i} r_i}
     =  g_2^{- \sum_{i=1}^{\ell} \alpha^{d_i} r_i}.
    \end{align*}
Note that the temporal decryption key is the same as the private key of the
above HIBE scheme.

\subsection{Security Analysis}

We first prove the security of the modified HIBE scheme. Note that the
selective KEM security model of the modified HIBE scheme can be easily
derived from the original selective KEM security model of HIBE by simply
incorporating a time period $T$ in private keys and the challenge ciphertext.
We omit the description of the security model.

\begin{theorem} \label{thm:bb-hibe-indcpa}
The above modified HIBE scheme is selectively secure under chosen plaintext
attacks if the $(3,N)$-MDHE assumption holds.
\end{theorem}

\begin{proof}
Suppose there exists an adversary $\mc{A}$ that attacks the above HIBE scheme
with a non-negligible advantage. A simulator $\mc{B}$ that solves the MDHE
assumption using $\mc{A}$ is given: a challenge tuple
    $D = \big( g_1, g_1^{a}, g_1^{a^2}, \ldots, g_1^{a^N}, g_1^{a^{N+2}},
    \ldots, g_1^{a^{2N}}, g_1^b, g_1^c \big)$ and $Z$
where $Z = Z_0 = g_3^{a^{N+1} bc}$ or $Z = Z_1 \in \G_3$. Then $\mc{B}$ that
interacts with $\mc{A}$ is described as follows:

\vs \noindent \tb{Init:} $\mc{A}$ initially submits a challenge identity
$ID^*|_{\ell^*} = (I_1^*, \ldots, I_{\ell^*}^*)$ and challenge time $T^*$. It
obtains $CID^*|_{\ell^*} = (CI_1^*, \ldots, CI_{\ell^*}^*)$ by calling
$\tb{EncodeCID}(ID^*|_{\ell^*})$.

\svs \noindent \tb{Setup:} $\mc{B}$ first chooses random exponents $\{
f'_{i,0} \}_{1 \leq i \leq L}, \{ f'_{i,j, k} \}_{1 \leq i \leq L, 1 \leq j
\leq l_1, k \in \bits}, h'_{0}, \{ h'_{j,k} \}_{1 \leq j \leq l_2, k \in
\bits}, \in \Z_p$. It implicitly sets $\alpha = a^{N+1}b$ and publishes the
public parameters $PP$ as
    \begin{align*}
    \vec{f}_{1,i} &= \Big(
        f_{1,i,0} = g_1^{f'_{i,0}} \big( \prod_{j=1}^{l_1} f_{1,i,j,CI^*[j]}
        \big)^{-1},~
        \big\{ f_{1,i,j,k} = \big( g_1^{a^N} \big)^{f'_{i,j,k}}
        \big\}_{1 \leq i \leq L, 1 \leq j \leq l_1, k \in \bits}
    \Big),~ \\
    \vec{h}_1 &= \Big(
        h_{1,0} = g_1^{h'_0} \big( \prod_{j=1}^{l_2} h_{1,j,T^*[j]} \big)^{-1},~
        \big\{ h_{1,j,k} = \big( g_1^b \big)^{h'_{j,k}}
        \big\}_{1 \leq j \leq l_2, k \in \bits}
    \Big),~ \\
    \Lambda &= e \big( e (g_1^{\alpha}, g_1^{\alpha^N}), g_1^b \big)
        = g_3^{\alpha^{N+1} b}.
    \end{align*}
For notational simplicity, we define $\Delta CI_i = \sum_{j=1}^{l_1}
(f'_{i,j,CI[j]} - f'_{i,j,CI^*[j]})$ and $\Delta T = \sum_{j=1}^{l_2}
(h'_{j,T[j]} - h'_{j,T^*[j]})$. We have $\Delta CI_i \not\equiv 0 \mod p$
except with negligible probability if $CI_i \neq CI^*_i$ since there exists
at least one index $j$ such that $f'_{i,j,CI[j]} \neq f'_{i,j,CI^*[j]}$ and
$\{ f'_{i,j,k} \}$ are randomly chosen. We also have $\Delta T \not\equiv 0
\mod p$ except with negligible probability if $T \neq T^*$.

\vs \noindent \tb{Phase 1:} $\mc{A}$ adaptively requests a polynomial number
of private key queries. If this is a private key query for a hierarchical
identity $ID|_\ell = (I_1, \ldots, I_\ell)$ and a time period $T$, then
$\mc{B}$ obtains $CID|_\ell = (CI_1, \ldots, CI_\ell)$ by calling
$\tb{EncodeCID}(ID|_\ell)$ and proceeds as follows.
\begin{itemize}
\item \tb{Case $ID_\ell \not\in \tb{Prefix}(ID^*|_{\ell^*})$}: In this
    case, we have $CI_\ell \neq CI^*_i$ for all $i$ by the property of the
    encoding function \cite{LeeP16}.
    It selects random exponents $r_1, \ldots, r_{\ell-1}, r'_\ell, r_{L+1}
    \in \Z_p$ and creates a private key $SK_{ID_{\ell},T}$ by implicitly
    setting $r_\ell = (-a / \Delta CI_\ell + r'_\ell) b$ as
    \begin{align*}
    & D_0   = e \big( (g_1^a)^{-f'_{\ell,0} / \Delta CI_\ell}
              F_1(CI_\ell)^{r'_\ell}, g_1^b \big) \cdot
              \prod_{i=1}^{\ell-1} F_{2,i}(CI_i)^{r_i} \cdot H_2(T)^{r_{L+1}},~ \\
    & \big\{ D_i = g_2^{r_i} \big\}_{1 \leq i \leq \ell-1},~
    D_\ell  = e \big( (g_1^a)^{-1 / \Delta CI_\ell} g_1^{r'_\ell}, g_1^b \big),~
    D_{L+1} = g_2^{r_{L+1}}.
    \end{align*}

\item \tb{Case $ID|_\ell = ID^*|_{\ell^*}$}: In this case, we have $T \neq
    T^*$. It selects random exponents $r_{1}, \ldots, r_{\ell}, r'_{L+1}
    \in \Z_p$ and creates a private key $SK_{ID_{\ell},T}$ by implicitly
    setting $r_{L+1} = (-a / \Delta T + r'_{L+1}) a^N$ as
    \begin{align*}
    & D_0   = e \big( (g_1^a)^{-h'_0 / \Delta T} H_1(T)^{r'_{L+1}},
              g_1^{a^N} \big) \cdot \prod_{i=1}^{\ell} F_{2,i}(CI_i)^{r_i},~ \\
    & \{ D_i = g_2^{r_i} \}_{1 \leq i \leq \ell},~
    D_{L+1} = e \big( (g_1^a)^{-1 / \Delta T} g_1^{r'_{L+1}}, g_1^{a^N} \big).
    \end{align*}
\end{itemize}

\noindent \tb{Challenge}: $\mc{B}$ creates the challenge ciphertext header
$CH^*$ by implicitly setting $s = c$ as
    \begin{align*}
    &   C_0 = g_1^c,~
        \big\{ C_i = ( g_1^c )^{f'_{i,0}} \big\}_{1 \leq i \leq \ell^*},~
        C_{L+1} = ( g_1^c )^{h'_0}
    \end{align*}
and the challenge session key $EK^* = Z$.

\svs \noindent \tb{Phase 2}: Same as Phase 1.

\svs \noindent \tb{Guess}: Finally, $\mc{A}$ outputs a guess $\mu' \in
\bits$. $\mc{B}$ also outputs $\mu'$.
\end{proof}

\begin{theorem} \label{thm:rhibe-hpu-srlind}
The above RHIBE scheme is SRL-IND secure if the $(3,N)$-MDHE assumption holds
where $N$ is the maximum number of child users in the system.
\end{theorem}

\begin{proof}
Suppose there exists an adversary $\mc{A}$ that attacks the above RHIBE scheme
with a non-negligible advantage. A meta-simulator $\mc{B}$ that solves the MDHE
assumption using $\mc{A}$ is given: a challenge tuple
    $D = \big( (p,\G_1, \G_2, \G_3), g_1, g_1^{a}, g_1^{a^2}, \ldots,
    g_1^{a^N}, g_1^{a^{N+2}}, \ldots, g_1^{a^{2N}}, g_1^b, g_1^c \big)$
    and $Z$
where $Z = Z_0 = g_3^{a^{N+1} bc}$ or $Z = Z_1 \in \G_3$. Note that a
challenge tuple
    $D_{BDHE} = \big( (p,\G_1, \G_2), g_1, g_1^{a}, g_1^{a^2}, \ldots,
    g_1^{a^N}, g_1^{a^{N+2}}, \ldots, g_1^{a^{2N}}, g_1^b \big)$
for the BDHE assumption can be derived from the challenge tuple $D$ of the
MDHE assumption. Let $\mc{B}_{HIBE}$be the simulator in the security proof of
Theorem \ref{thm:bb-hibe-indcpa} and $\mc{B}_{PKBE}$ be a simulator in
security proof of Theorem \ref{thm:bgw-pkbe-indcpa}. Then $\mc{B}$ that
interacts with $\mc{A}$ is described as follows:

\vs \noindent \textbf{Init:} $\mc{A}$ initially submits a challenge identity
$ID^*|_{\ell^*} = (I^*_1, \ldots, I^*_{\ell^*})$, a challenge time period
$T^*$, and a revoked identity set $R^* = (R^*_{ID_0}, \ldots,
R^*_{ID_{L-1}})$ at the time period $T^*$.
It first sets a state $ST$ and a revocation list $RL$ as empty one. For each
$ID_k \in \{ ID^* \} \cup R^*$, it selects a unique index $d_i \in \mc{N}$
such that $(-, d_i) \notin ST_{ID_{k-1}}$ and adds $(ID_k, d_i)$ to
$ST_{ID_{k-1}}$. Let $RI^* = (RI^*_{ID_0}, \ldots, RI^*_{ID_{\ell^*-1}})
\subseteq \mc{N}$ be the revoked index set of $R^*$ at the time $T^*$ and
$SI^* = (SI^*_{ID_1}, \ldots, SI^*_{ID_{\ell^*-1}}$) be the non-revoked index
set at the time $T^*$ such that $SI^*_{ID_x} = \mc{N} \setminus RI^*_{ID_x}$.

\svs \noindent \textbf{Setup:} $\mc{B}$ submits $ID^*_{\ell^*}$ and $T^*$ to
$\mc{B}_{HIBE}$ and receives $PP_{HIBE}$. It also submits $SI^*_{\epsilon}$
to $\mc{B}_{PKBE}$ and receives $PP_{PKBE}$.
Note that $\{ ID^*_i \notin R_{ID^*_{i-1}}\}_{1 \leq i < x}$ and $\{ ID^*_i
\in R_{ID^*_{i-1}}\}_{x \leq i \leq \ell^*}$. $\mc{B}$ first chooses random
exponents $\theta_{ID^*_0}, \ldots, \theta_{ID^*_{\ell^*-1}},
\hat{\beta}_{ID^*_0}, \ldots, \hat{\beta}_{ID^*_{x-1}}, \beta_{ID^*_{x}},
\ldots, \beta_{ID^*_{\ell^*}} \in \Z_p$.
It implicitly sets $\alpha = a, \beta_{\epsilon} = b + \hat{\beta}_{ID_0},
\{\beta_{ID_i} = b + \hat{\beta}_{ID_i}\}_{1 \leq i < x}, \{\gamma_{ID^*_x} =
\theta_{ID^*_x} - \sum_{j \in SI_{ID^*_x}} a^{N+1-j}\}_{0 \leq x \leq
\ell^*-1}$ and publishes the public parameters $PP$ as
    \begin{align*}
    PP = \Big(
    &   PP_{HIBE},~ PP_{PKBE},~
        g_2^{\alpha^{N+1}} = e( g_1^a, g_2^{a^N} ),~
        g_2^{\beta_{\epsilon}} = g_2^{b + \hat{\beta}_{ID_0}}
            = e( g_1,  g_1^b \cdot g_1^{\hat{\beta}_{ID_0}}),~ \\
    &   \Omega = g_3^{\alpha^{N+1} \beta_{\epsilon}}
            = g_3^{a^{N+1} (b + \hat{\beta}_{ID_0})}
            = e \big( e (g_1^a, g_1^{a^N}), g_1^b \big) \cdot
              e \big( e (g_1^a, g_1^{a^N}), g_1^{\hat{\beta}_{ID_0}} \big)
    \Big).
    \end{align*}

\vs \noindent \textbf{Phase 1:} $\mc{A}$ adaptively requests a polynomial
number of private key, update key, and decryption key queries.

\svs \noindent If this is a private key query for an identity $ID|_{\ell} =
(I_1, \ldots, I_{\ell})$, then $\mc{B}$ proceeds as follows: Note that
$SK_{ID|_{\ell}} = \big( \{ d_i, LSK_i \}_{i=1}^{\ell} \big)$ where $LSK_i =
(K_{i,0}, K_{i,1}, R_{i,0}, R_{i,1})$.

\begin{itemize}
\item \tb{Case} $ID|_{\ell-1} \notin \tb{Prefix}(ID^*|_{\ell^*})$: It first
    normally generates a state $ST_{ID|_{\ell-1}}$ and adds a tuple
    $(ID|_{\ell}, d_{\ell})$ to $ST_{ID|_{\ell-1}}$ where $d_\ell$ is an
    index for $ID|_{\ell}$.
    It obtains $SK_{ID|_{\ell-1}}$ by requesting an RHIBE private key query
    for $ID|_{\ell-1}$ to $\mc{B}_{HIBE}$. Next, it simply generates
    $SK_{ID|_{\ell}}$ by running $\tb{RHIBE.GenKey} (ID|_{\ell},
    SK_{ID|_{\ell-1}}, ST_{ID|_{\ell-1}}, PP)$.

\item \tb{Case} $ID|_{\ell-1} \in \tb{Prefix}(ID^*|_{\ell^*})$: Note that
    $ID_0 = ID_{\ell-1}$ is included in this case. We have $ID_{\ell} =
    (I^*_1, \ldots, I^*_{\ell-1}, I_{\ell})$, $\{ ID^*_i \notin
    R_{ID^*_{i-1}} \}_{1 \leq i < x}$ and $\{ ID^*_i \in
    R_{ID^*_{i-1}}\}_{x \leq i \leq \ell-1}$.
    \begin{itemize}
    \item \tb{Case} $ID|_{\ell} \in R^*_{ID^*_{\ell-1}}$: In this case,
        it first retrieves a tuple $(ID_{\ell}, d_{\ell})$ from
        $ST_{ID^*_{\ell-1}}$ where the index $d_{\ell}$ is associated
        with $ID_{\ell}$. Note that the tuple $(ID_{\ell}, d_{\ell})$
        exists since all identities in $R^*_{ID^*_{\ell-1}}$ were added
        to $ST_{ID^*_{\ell-1}}$ in the initialization step.

    If $1 \leq i < x$, the simulator can use the cancellation technique
    by using the random blind element.
    It recalls random exponents $\theta_{ID^*_{0}}, \ldots,
    \theta_{ID^*_{x-2}}, \hat{\beta}_{ID^*_{0}}, \ldots,
    \hat{\beta}_{ID^*_{x-1}}$ and selects random exponents $r_{1,1},
    \ldots, r_{x-1,1}, \hat{r}_{1,2}, \ldots, \hat{r}_{x-1,2} \in \Z_p$
    and creates level private keys by implicitly setting $\{r_{i,2} =
    \hat{r}_{i,2} - a^{N+1} \}_{1 \leq i < x}$ as
    \begin{align*}
    \Big\{
    LSK_i = \Big(
        K_{i,0} &= \big( g_1^{a^{d_{i}}} \big)^{\theta_{ID^*_{i-1}}}
                \big( \prod_{j \in SI^*_{ID_{i-1}}, j \neq d_{i}}
                g_1^{a^{N+1-j+d_{i}}} \big)^{-1} F_{1,i}(I_i)^{-r_{i,1}}
                \cdot g_1^{-\hat{r}_{i,2}},~
        K_{i,1} = g_1^{-r_{i,1}},~\\
        R_{i,0} &= e(g_1^b \cdot g_1^{\hat{\beta}_{ID^*_{i-1}}}, g_1),~
        R_{i,1} = (g_2^{a^{N+1}})^{\hat{\beta}_{ID^*_{i}}} \cdot
                 e(g_1^b, g_1^{\hat{r}_{i,2}}) \cdot
                 g_2^{\hat{r}_{i,2}\hat{\beta}_{ID^*_{i-1}}} \cdot
                (g_2^{a^{N+1}})^{-\hat{\beta}_{ID^*_{i-1}}}
                  \Big)\Big\}_{1 \leq i < x}.
    \end{align*}
    If $i = x$, the simulator can use the partitioning technique of
    Boneh et al.\cite{BonehGW05}.
    It recalls random exponents $\gamma_{ID^*_{i-1}},
    \hat{\beta}_{ID^*_{i-1}}, \beta_{ID^*_{i}}$ and obtains $SK_{BE} =
    K_{BE}$ by requesting an PKBE private key query for $(d_{i},
    \gamma_{ID^*_{i-1}})$.
    Next, it selects random exponents $r_{i,1}, r_{i,2} \in \Z_p$ and
    creates a level private key as
    \begin{align*}
    LSK_i = \Big(
        K_{i,0} &= K_{BE} \cdot F_{1,i}(I_i)^{-r_{i,1}} \cdot g_1^{-r_{i,2}},~
        K_{i,1}  = g_1^{-r_{i,1}},~ \\
        R_{i,0} &= e( g_1^{\beta_{ID^*_{i}}}, g_1),~
        R_{i,1}  = e(g_1^{a^{N}}, g_1^{a})^{\beta_{ID_{i}}}
                   \cdot e(g_1^{b} \cdot g_1^{\hat{\beta}_{ID^*_{i-1}}},
                   g_1^{r_{i,2}})
    \Big).
    \end{align*}
    If $x < i \leq \ell-1$, the simulator can use the partitioning
    technique of Boneh et al.\cite{BonehGW05}.
    It recalls random exponents $\gamma_{ID^*_{x}}, \ldots,
    \gamma_{ID^*_{\ell-2}}, \beta_{ID^*_{x}}, \ldots,
    \beta_{ID^*_{\ell-1}}$ and obtains $\{SK_{BE,i} = K_{BE,i} \}_{x < i
    \leq \ell-1}$ by requesting an PKBE private key query for $\{ (d_{i},
    \gamma_{ID^*_{i-1}}) \}_{x < i \leq \ell-1}$.
    Next, it selects random exponents $r_{x+1,1}, \ldots, r_{\ell-1,1},
    r_{x+1,2}, \ldots, r_{\ell-1,2} \in \Z_p$ and creates level private
    keys as
    \begin{align*}
    \{
    LSK_i = \Big(
        K_{i,0} &= K_{BE} \cdot F_{1,i}(I_i)^{-r_{i,1}}
                   \cdot g_1^{-r_{i,2}},~
        K_{i,1}  = g_1^{-r_{i,1}},~ \\
        R_{i,0} &= e(g_1^{\beta_{ID_{i}}}, g_1),~
        R_{i,1}  = e(g_1^{a^{N}}, g_1^{a})^{\beta_{ID_{i}}}
                   \cdot e(g_1^{\beta_{ID^*_{i-1}}}, g_1^{r_{i,2}})
    \Big)
    \}_{x < i \leq \ell-1}.
    \end{align*}

    If $i = \ell$, the simulator can use the partitioning technique of
    Boneh et al. \cite{BonehGW05}.
    It recalls a random exponent $\gamma_{ID^*_{i-1}}$
    and obtains $SK_{BE} = K_{BE}$ by requesting an PKBE private key
    query for $(d_{i}, \gamma_{ID^*_{i-1}})$.
    Next, it selects a random exponent $r_{i,1} \in \Z_p$ and creates a
    level private key as
    \begin{align*}
    LSK_i = \Big(
        K_{i,0} &= K_{BE} \cdot F_{1,i}(I_i)^{-r_{i,1}}
                \cdot g_1^{-r_{i,2}},~
        K_{i,1}  = g_1^{-r_{i,1}},~\\
            R_{i,0} &= 1_{\G_2},~
            R_{i,1}  = 1_{\G_2}
                   \Big).
    \end{align*}
    \item \tb{Case} $ID_{\ell} \notin R^*_{ID^*_{\ell-1}}$: In this case,
        we have $\{ID_i \notin R_{ID_{i-1}}\}_{1 \leq i \leq \ell}$ and
        $\ell < x$, since, if a parents identity is revoked $ID_i \in
        R_{ID_{i-1}}$, then a child identity should be revoked $ID_{i+1}
        \in R_{ID_{i}}$. It first selects an index $d_{\ell} \in \mc{N}$
        such that $(-,d_{\ell}) \notin ST_{ID^*_{\ell-1}}$ and adds
        $(ID_{\ell}, d_{\ell})$ to $ST_{ID^*_{\ell-1}}$.

    If $1 \leq i < \ell-1$, the simulator can use the cancellation
    technique by using the random blind element.
    It recalls random exponents $\theta_{ID^*_{0}}, \ldots,
    \theta_{ID^*_{\ell-2}}, \hat{\beta}_{ID^*_{0}}, \ldots,
    \hat{\beta}_{ID^*_{\ell-1}}$ and selects random exponents $r_{1,1},
    \ldots, r_{\ell-1,1}, \hat{r}_{1,2}, \ldots, \hat{r}_{\ell-1,2} \in
    \Z_p$ and creates level private keys by implicitly setting $\{r_{i,2}
    = \hat{r}_{i,2} - a^{N+1}, \gamma_{ID^*_{i-1}} = \theta_{ID^*_{i-1}}
    - \sum_{j \in SI_{ID^*_{i-1}}} a^{N+1-j}\}_{1 \leq i < \ell}$ as
    \begin{align*}
    \Big\{
    LSK_i = \Big(
        K_{i,0} &= \big( g_1^{a^{d_{i}}} \big)^{\theta_{ID^*_{i-1}}}
                \big( \prod_{j \in SI^*_{ID_{i-1}}, j \neq d_{i}}
                g_1^{a^{N+1-j+d_{i}}} \big)^{-1} F_{1,i}(I_i)^{-r_{i,1}}
                \cdot g_1^{-\hat{r}_{i,2}},~
        K_{i,1} = g_1^{-r_{i,1}},~\\
        R_{i,0} &= e(g_1^b \cdot g_1^{\hat{\beta}_{ID^*_{i-1}}}, g_1),~
        R_{i,1} = (g_2^{a^{N+1}})^{\hat{\beta}_{ID^*_{i}}} \cdot
                 e(g_1^b, g_1^{\hat{r}_{i,2}}) \cdot
                 g_2^{\hat{r}_{i,2}\hat{\beta}_{ID^*_{i-1}}} \cdot
                (g_2^{a^{N+1}})^{-\hat{\beta}_{ID^*_{i-1}}}
                  \Big)\Big\}_{1 \leq i < \ell}.
    \end{align*}

    If $i = \ell$, the simulator can use the partitioning technique of
    Boneh and Boyen \cite{BonehB04e}. We have $I_{i} \neq I^*_{i}$ from
    the restriction of Definition \ref{def:rhibe-hpu-srlind}. It selects
    a random exponent $r'_{i,1} \in \Z_p$ and creates a level private key
    by implicitly setting $r_{i,1} = -a / \Delta I_{i} + r'_{i, 1}$ as
    \begin{align*}
    LSK_i = \Big(
         K_{i,0} &= g_1^{a^d \theta_{ID_{i-1}}} \prod_{j \in SI^*_{ID_
                {i-1}} \setminus \{ d_{i} \} } g_1^{-a^{N+1-j+d_
                {i}}} (g_1^a)^{f'_0 / \Delta I_{i}} F_{1,i}
                (I_{i})^{-r'_{i,1}},~
         K_{i,1} = (g_1^a)^{-1 / \Delta I_{i}} g_1^{r'_{i,1}},~\\
            R_{i,0} &= 1_{\G_2},~
            R_{i,1}  = 1_{\G_2}
         \Big).
    \end{align*}
    \end{itemize}
\end{itemize}

\noindent If this is an update key query for an identity $ID|_{\ell-1} =
(I_1, \ldots, I_{\ell-1})$ and a time period $T$, then $\mc{B}$ defines a
revoked identity set $R_{ID_{\ell-1}}$ at the time $T$ from
$RL_{ID_{\ell-1}}$ and proceeds as follows: Note that $UK_{T,R_{ID_{\ell-1}}}
= \big( \{ SI_{ID_i}, LUK_i \}_{_i=0}^{\ell-1} \big)$ where $LUK_i =
(U_{i,0}, U_{i,1}, U_{i,2})$.
We assume that $(I_1 = I^*_1), \ldots, (I_y = I^*_y), (I_{y+1} \neq
I^*_{y+1}), \ldots, \lb (I_{\ell-1} \neq I^*_{\ell-1})$ where $1 \leq y \leq
\ell-1$. And, we have $\{ ID^*_i \notin R_{ID^*_{i-1}} \}_{1 \leq i < x}$ and
$\{ ID^*_i \in R_{ID^*_{i-1}} \}_{x \leq i \leq \ell-1}$.

\begin{itemize}
\item \tb{Case} $T \neq T^*$ : It first sets a revoked index set
    $RI_{ID_{\ell-1}}$ of $R_{ID_{\ell-1}}$ by using $ST_{ID_{\ell-1}}$. It
    also sets $SI_{ID_{\ell-1}} = \mc{N}_{\ell} \setminus
    RI_{ID_{\ell-1}}$.

    If $0 \leq i \leq x$, the simulator can use the partitioning technique
    of Boneh and Boyen\cite{BonehB04e}. It recalls a random exponent
    $\hat{\beta}_{ID^*_{0}},\ldots, \hat{\beta}_{ID^*_{x}},
    \theta_{ID^*_{0}},\ldots,\theta_{ID^*_{x}}$ and selects a random
    exponent $r'_{0},\ldots,r'_{x} \in Z_p$. It creates level update keys
    by implicitly setting $\{r_{i} = - (-\sum_{j \in SI^*_{ID^*_{i}}
    \setminus SI_{ID_{i}}}$ $a^{N+1-j} + \sum_{j \in SI_{ID_{i}} \setminus
    SI^* _{ID^*_{i}}} a^{N+1-j}) / \Delta T + r'_{i}\}_{0 \leq i \leq x}$
    as
    \begin{align*}
    \Big\{
    LUK_i = \Big(
        U_{i,0} &= g_1^{b} \cdot g_1^{\hat{\beta}_{ID^*_{i}}},~\\
        U_{i,1} &= (g_1^b \cdot g_1^{\hat{\beta}_{ID^*_{i}}})^{\theta_{ID^*_{i}}}
               \Big( \prod_{j \in SI^*_{ID^*_{i}} \setminus SI_{ID_{i}}}
                     g_1^{-a^{N+1-j}}
                     \prod_{j \in SI_{ID_{i}} \setminus SI^*_{ID^*_{i}}}
                     g_1^{a^{N+1-j}}
               \Big)^{-h'_0 / \Delta T}
               H_1(T)^{r'_{i}},~ \\
        U_{i,2} &= \Big( \prod_{j \in SI^*_{ID^*_{i}} \setminus SI_{ID_{i}}}
                     g_1^{-a^{N+1-j}}
                     \prod_{j \in SI_{ID_{i}} \setminus SI^*_{ID^*_{i}}}
                     g_1^{a^{N+1-j}}
           \Big)^{-1 / \Delta T} g_1^{r'_{i}}
          \Big) \Big\}_{0 \leq i \leq x}.
    \end{align*}
    If $x < i \leq y$, it recalls a random exponent $\beta_{ID^*_{x+1}},
    \ldots, \beta_{ID^*_{y}}, \theta_{ID^*_{x+1}}, \ldots,
    \theta_{ID^*_{y}}$ and selects a random exponent $r_{x+1}, \ldots,
    r_{y} \in Z_p$. It creates a level update key as
    \begin{align*}
    \Big\{
    LUK_i = \Big(
        U_{i,0} &= g_1^{\beta_{ID^*_{i}}},~\\
        U_{i,1} &= \Big( g_1^{\theta_{ID_i}} \big( \prod_{j \in SI^*_{ID_i}}
                 g_1^{a^{N+1-j}} \big)^{-1} \prod_{j \in SI_{ID_i}}
                 g_1^{a^{N+1-j}} \Big)^{\beta_{ID^*_i}}
                 H_1(T)^{r_{i}}\\
        U_{i,2} &= g_1^{r_{i}} \Big) \Big\}_{x < i \leq y}.
    \end{align*}
    If $y < i \leq \ell-1$, the simulator can normally generate level
    private keys. It selects random exponents $r_{y+1},\ldots, r_{\ell-1},
    \gamma_{ID_{y+1}}, \ldots, \gamma_{ID_{\ell-1}},
    \beta_{ID_{y+1}},\ldots, \beta_{ID_{\ell-1}} \in \Z_p$ and creates
    level private keys as
    \begin{align*}
    \Big\{
    LUK_i = \Big(
        U_{i,0} = g_1^{\beta_{ID_i}},~
        U_{i,1} = \big( g_1^{\gamma_{ID_i}} \prod_{j \in SI_{ID_i}}
                      g_1^{a^{N+1-j}}
                      \big)^{\beta_{ID_i}} H_1(T)^{r_{i}},~
        U_{i,2} = g_1^{-r_{i}}
                    \Big) \Big\}_{y < i \leq \ell-1}.
    \end{align*}
\item \tb{Case} $T = T^*$ : We have $R = R^*$.
    For each $ID_{i} \in R^*_{ID^*_{i-1}}$, it adds $(ID_{i}, T^*)$ to
    $RL_{ID^*_{i-1}}$ if $(ID_{i}, T') \notin RL_{ID^*_{i-1}}$ for any
    $T' \leq T^*$.

    If $0 \leq i \leq x$, the simulator can use the partitioning technique
    of Boneh et al. \cite{BonehGW05}. It recalls random exponents
    $\hat{\beta}_{ID^*_{0}}, \ldots, \hat{\beta}_{ID^*_x}, \theta_{ID^*_0},
    \ldots, \theta_{ID^*_x}$ and selects a random exponent $r_{0}, \ldots,
    r_{x} \in \Z_p$. It creates level update keys as
    \begin{align*}
    \Big\{
    LUK_i = \Big(
        U_{i,0}  = g_1^b \cdot g_1^{\hat{\beta}_{ID^*_i}},~
        U_{i,1}  = (g_1^b \cdot g_1^{\hat{\beta}_{ID^*_i}})^{\theta_{ID^*_i}}
                   H_1(T^*)^{r_{i,0}},~
        U_{i,2}  = g_1^{-r_{i,0}}
            \Big) \Big\}_{0 \leq i \leq x}.
    \end{align*}
    If $x < i \leq y$, the simulator can use the partitioning technique of
    Boneh et al. \cite{BonehGW05}. It recalls random exponents
    $\beta_{ID^*_{x+1}}, \ldots, \beta_{ID^*_y}, \theta_{ID^*_{x+1}},
    \ldots, \theta_{ID^*_y}$ and selects a random exponent $r_{x+1},
    \ldots, r_{y} \in \Z_p$. It creates level update keys as
    \begin{align*}
    \Big\{
    LUK_i = \Big(
        U_{i,0}  = g_1^{\beta_{ID^*_i}},~
        U_{i,1}  = (g_1^{\beta_{ID^*_i}})^{\theta_{ID^*_i}}
                   H_1(T^*)^{r_{i,0}},~
        U_{i,2}  = g_1^{-r_{i,0}}
            \Big) \Big\}_{x < i \leq y}.
    \end{align*}
    If $y < i \leq \ell-1$, the simulator can normally generate level
    private keys. It selects random exponents $r_{y+1},\ldots, r_{\ell-1},
    \gamma_{ID_{y+1}}, \ldots, \gamma_{ID_{\ell-1}}, \beta_{ID_{y+1}},
    \ldots, \beta_{ID_{\ell-1}} \in \Z_p$ and creates level private keys as
    \begin{align*}
    \Big\{
    LUK_i = \Big(
        U_{i,0} = g_1^{\beta_{ID_i}},~
        U_{i,1} = \big( g_1^{\gamma_{ID_i}} \prod_{j \in SI_{ID_i}}
                      g_1^{a^{N+1-j}}
                      \big)^{\beta_{ID_i}} H_1(T)^{r_{i}},~
        U_{i,2} = g_1^{-r_{i}}
                    \Big) \Big\}_{y < i \leq \ell-1}.
    \end{align*}
\end{itemize}

\noindent If this is a decryption key query for an identity $ID = (I_1,
\ldots, I_{\ell})$  and a time period $T$, then $\mc{B}$ proceeds as follows:
It requests an HIBE private key for $ID$ and $T$ to $\mc{B}_{HIBE}$ and
receives $SK_{HIBE,ID,T}$. Next, it sets the decryption key $DK_{ID,T} =
SK_{HIBE,ID,T}$.

\svs \noindent \textbf{Challenge}: $\mc{A}$ submits two challenge messages
$M_0^*, M_1^*$. $\mc{B}$ chooses a random bit $\delta \in \bits$ and proceed
as follows:
It requests the challenge ciphertext for $ID^*$ and $T^*$ to $\mc{B}_{HIBE}$
and receives $CH_{HIBE, ID^*, T^*}$. Next, it sets the challenge ciphertext
$CT_{ID^*, T^*} = \big( C = Z \cdot  e( e(g_1^a, g_1^{a^N}),
g_1^c)^{\hat{\beta}_{ID_0}} \cdot M_{\delta}^*,~ CH_{HIBE, ID^*, T^*} \big)$.

\svs \noindent \textbf{Phase 2}: Same as Phase 1.

\svs \noindent \textbf{Guess}: Finally, $\mc{A}$ outputs a guess $\delta' \in
\bits$. $\mc{B}$ outputs $0$ if $\delta = \delta'$ or $1$ otherwise.

\vs To finish the proof, we first show that the distribution of the
simulation is correct from Lemma \ref{lem:rhibe-hpu-dist}.
This completes our proof.
\end{proof}

\begin{lemma} \label{lem:rhibe-hpu-dist}
The distribution of the above simulation is correct if $Z = Z_0$, and the
challenge ciphertext is independent of $\delta$ in the adversary's view if $Z
= Z_1$.
\end{lemma}

\begin{proof}

We show that the distribution of private keys is correct. In case of
$ID_{\ell} \in R^*_{ID^*_{\ell-1}}$ and $ID_{\ell-1} \in
\tb{Prefix}(ID^*_{\ell^*})$, we have that the private key is correctly
distributed from the setting $\{r_{i, 2} = \hat{r}_{i,2} - a^{N+1}\}_{1 \leq
i < x}$, $\{\beta_{ID^*_{i}} = b + \hat{\beta}_ {ID^*_{i}}\}_{0 \leq i < x}$,
and $\{\gamma_{ID^*_{i-1}} = \theta_{ID^*_{i-1}} - \sum_{j \in
SI_{ID^*_{i-1}}} a^{N+1-j}\}_{1 \leq i < x}$ as the following equation
    \begin{align*}
    \{
    K_{i, 0} &= g_1^{\alpha^{d_{i}} \gamma_{ID^*_{i-1}}} F_{1,i}(I_i)
                        ^{-r_{i,1}} \cdot g_1^{-r_{i,2}}
              = g_1^{a^{d_{i}} \theta_{ID^*_{i-1}}} \prod_{j \in
                     SI^*_{ID_{i-1}}} g^{-a^{N+1-j+d_{i}}}
                     F_{1,i}(I_i)^{-r_{i,1}} \cdot g_1^{-\hat{r}_{i,2} + a^{N+1}}\\
                &= g_1^{a^{d_{i}} \theta_{ID^*_{i-1}}} \prod_{j \in SI^*_{ID_{i-1}}
                       \setminus \{ d_{i} \} } g_1^{-a^{N+1-j+d_{i}}}
                       \cdot g_1^{-a^{N+1}} F_{1,i}(I_i)^{-r_{i,1}} \cdot
                       g_1^{-\hat{r}_{i,2}} \cdot g_1^{a^{N+1}} \\
                &= \big( g_1^{a^{d_{i}}} \big)^{\theta_{ID^*_{i-1}}}
                \big( \prod_{j \in SI^*_{ID_{i-1}} \setminus \{ d_{i} \}}
                g_1^{a^{N+1-j+d_{i}}} \big)^{-1} F_{1,i}(I_i)^{-r_{i,1}}
                g_1^{-\hat{r}_{i,2}},~\\
    R_{i, 0} &= g_2^{\beta_{ID^*_{i-1}}} = g_2^{b + \hat{\beta}_{ID^*_{i-1}}}
              = e(g_1^b \cdot g_1^{\hat{\beta}_{ID^*_{i-1}}}, g_1),~\\
    R_{i, 1} &= g_2^{\alpha^{N+1}\beta_{ID_i}} \cdot g_2^{\beta_{ID_{i-1}}r_{i,2}}
              = g_2^{a^{N+1}(b + \hat{\beta}_{ID^*_{i}})} \cdot
                g_2^{(b + \hat{\beta}_{ID^*_{i-1}})(\hat{r}_{i,2} - a^{N+1})} \\
             &= g_2^{a^{N+1}b + a^{N+1}\hat{\beta}_{ID^*_{i}}} \cdot
                g_2^{\hat{r}_{i,2}b - a^{N+1}b + \hat{r}_{i,2}\hat{\beta}_{ID^*_{i-1}}
                     - a^{N+1}\hat{\beta}_{ID^*_{i-1}}} \\
             &= (g_2^{a^{N+1}})^{\hat{\beta}_{ID^*_{i}}} \cdot
                 e(g_1^b, g_1^{\hat{r}_{i,2}}) \cdot
                 g_2^{\hat{r}_{i,2}\hat{\beta}_{ID^*_{i-1}}} \cdot
                (g_2^{a^{N+1}})^{-\hat{\beta}_{ID^*_{i-1}}} \}_{1 \leq i < x}.
    \end{align*}
    \begin{align*}
    \{
    R_{i, 1} &= g_2^{\alpha^{N+1}\beta_{ID_i}} \cdot g_2^{\beta_{ID_{i-1}}r_{i,2}}
              = e(g_1^{\alpha^{N}}, g_1^{\alpha})^{\beta_{ID_i}} \cdot
                e(g_1^{\beta_{ID_{i-1}}}, g_1^{r_{i,2}})
              = e(g_1^{a^{N}}, g_1^{a})^{\beta_{ID_i}} \cdot
                e(g_1^{b} \cdot g_1^{\hat{\beta}_{ID^*_{i-1}}}, g_1^{r_{i,2}})
                \}_{i = x}.
    \end{align*}
In case of $ID_{\ell} \notin R^*_{ID^*_{\ell-1}}$ and $ID_{\ell-1} \in
\tb{Prefix} (ID^*_{\ell^*})$, we have that the private key is correctly
distributed from the setting $\{r_{i, 2} = \hat{r}_{i,2} - a^{N+1}\}_{1 \leq
i < \ell}$, $\{\beta_{ID^*_{i}} = b + \hat{\beta}_ {ID^*_{i}}\}_{0 \leq i <
\ell}$, $r_{\ell, 1} = -a / \Delta I_{\ell} + r'_{\ell, 1}$, and
$\{\gamma_{ID^*_{i-1}} = \theta_{ID^*_{i-1}} - \sum_{j \in SI_{ID^*_{i-1}}}
a^{N+1-j}\}_{0 \leq i < \ell}$ as the following equation
    \begin{align*}
    \{
    K_{i, 0} &= g_1^{\alpha^{d_{i}} \gamma_{ID^*_{i-1}}} F_{1,i}(I_i)
                        ^{-r_{i,1}} \cdot g_1^{-r_{i,2}}
              = g_1^{a^{d_{i}} \theta_{ID^*_{i-1}}} \prod_{j \in
                     SI^*_{ID_{i-1}}} g^{-a^{N+1-j+d_{i}}}
                     F_{1,i}(I_i)^{-r_{i,1}} \cdot g_1^{-\hat{r}_{i,2} + a^{N+1}}\\
                &= g_1^{a^{d_{i}} \theta_{ID^*_{i-1}}} \prod_{j \in SI^*_{ID_{i-1}}
                       \setminus \{ d_{i} \} } g_1^{-a^{N+1-j+d_{i}}}
                       \cdot g_1^{-a^{N+1}} F_{1,i}(I_i)^{-r_{i,1}} \cdot
                       g_1^{-\hat{r}_{i,2}} \cdot g_1^{a^{N+1}} \\
                &= \big( g_1^{a^{d_{i}}} \big)^{\theta_{ID^*_{i-1}}}
                \big( \prod_{j \in SI^*_{ID_{i-1}} \setminus \{ d_{i} \}}
                g_1^{a^{N+1-j+d_{i}}} \big)^{-1} F_{1,i}(I_i)^{-r_{i,1}}
                g_1^{-\hat{r}_{i,2}},~\\
    R_{i, 0} &= g_2^{\beta_{ID^*_{i-1}}} = g_2^{b + \hat{\beta}_{ID^*_{i-1}}}
              = e(g_1^b \cdot g_1^{\hat{\beta}_{ID^*_{i-1}}}, g_1),~\\
    R_{i, 1} &= g_2^{\alpha^{N+1}\beta_{ID_i}} \cdot g_2^{\beta_{ID_{i-1}}r_{i,2}}
              = g_2^{a^{N+1}(b + \hat{\beta}_{ID^*_{i}})} \cdot
                g_2^{(b + \hat{\beta}_{ID^*_{i-1}})(\hat{r}_{i,2} - a^{N+1})} \\
             &= g_2^{a^{N+1}b + a^{N+1}\hat{\beta}_{ID^*_{i}}} \cdot
                g_2^{\hat{r}_{i,2}b - a^{N+1}b + \hat{r}_{i,2}\hat{\beta}_{ID^*_{i-1}}
                     - a^{N+1}\hat{\beta}_{ID^*_{i-1}}} \\
             &= (g_2^{a^{N+1}})^{\hat{\beta}_{ID^*_{i}}} \cdot
                 e(g_1^b, g_1^{\hat{r}_{i,2}}) \cdot
                 g_2^{\hat{r}_{i,2}\hat{\beta}_{ID^*_{i-1}}} \cdot
                (g_2^{a^{N+1}})^{-\hat{\beta}_{ID^*_{i-1}}} \}_{1 \leq i < \ell}.
    \end{align*}
    \begin{align*}
    \{
    K_{i, 0} &= g_1^{\alpha^{d_{ID_{i-1}}} \gamma_{ID_{x-1}}} F_{1,i}(I_i)
               ^{-r_{i,1}}
             = g_1^{a^{d_{i}} \theta_{ID_{i-1}}} \prod_{j \in SI^*
                  _{ID_{i-1}}}
               g^{-a^{N+1-j+d_{i}}}
               \big( f_{1,0} \prod_{j=1}^l f_{1,i,I_i[j]} \big)^{-r_{i,1}} \\
            &= g_1^{a^{d_{i}} \theta_{ID_{i-1}}} \prod_{j \in SI^*
                  _{ID_{i-1}}
               \setminus \{ d_{i} \} } g_1^{-a^{N+1-j+d_{i}}}
               \cdot g_1^{-a^{N+1}} \big( g_1^{f'_0} g_1^{a^N \Delta I_i}
               \big)^{a / \Delta I_1 - r'_{{i-1},1}} \\
            &= g_1^{a^{d_{i}} \theta_{ID_{i-1}}} \prod_{j \in SI^*_
                  {ID_{i-1}} \setminus
               \{ d_{i} \} } g_1^{-a^{N+1-j+d_{i}}}
               (g_1^a)^{f'_0 / \Delta
               I_i} F_{1,i}(I_i)^{-r'_{i,1}},~\\
    K_{i, 1} &= g_1^{r_{i,1}}
             = (g_1^a)^{-1 / \Delta I_i} g_1^{r'_{i,1}} \}_{i = \ell}.
    \end{align*}

Next, we show that the distribution of update keys is correct. In case of $T
\neq T^*$, we have that the update key is correctly distributed from the
setting $\{r_{i} = - (-\sum_{j \in SI^*_{ID^*_{i}} \setminus SI_{ID_{i}}}$
$a^{N+1-j} + \sum_{j \in SI_{ID_{i}} \setminus SI^*_{ID^*_{i}}} a^{N+1-j}) /
\Delta T + r'_{i}\}_{0 \leq i \leq x}$, $\{\beta_{ID^*_{i}} = b +
\hat{\beta}_{ID^*_{i}}\} _{0 \leq i \leq x}$, and $\{\gamma_{ID^*_{i}} =
\theta_{ID^*_{i}} - \sum_{j \in SI_{ID^*_{i}}} a^{N+1-j}\}_{1 \leq i \leq y}$
as the following equation
    \begin{align*}
    \{
    U_{i,0} = &g_1^{\beta_{ID^*_i}} = g_1^b \cdot g_1^{\hat{\beta}_{ID^*_{i}}}, \\
    U_{i,1} = &\big( g_1^{\gamma_{ID_i}} \prod_{j \in SI_{ID_i}}
                 g_1^{\alpha^{N+1-j}} \big)^{\beta_{ID_{i-1}}} H_1(T)^{r_{i+1,0}}\\
            = &\Big( g_1^{\theta_{ID_i}} \big( \prod_{j \in SI^*_{ID_i}}
                 g_1^{a^{N+1-j}} \big)^{-1} \prod_{j \in SI_{ID_i}}
                 g_1^{a^{N+1-j}} \Big)^{b + \hat{\beta}_{ID^*_i}}
                 \big( h_{1,0} \prod_{k=1}^t h_{1,k,T[k]} \big)^{r_{i+1,0}} \db \\
            = &(g_1^{b + \hat{\beta}_{ID^*_i}})^{\theta_{ID_i}}
                 \Big( \prod_{j \in SI^*_{ID_i}
                 \setminus SI_{ID_i}} g_1^{-a^{N+1-j}}
                 \prod_{j \in SI_{ID_i} \setminus SI^*_{ID_i}}
                 g_1^{a^{N+1-j}} \Big)^{b + \hat{\beta}_{ID^*_i}}\\
                &\cdot
                 \big( g_1^{h'_0} g_1^{b \Delta T} \big)^{-( -\sum_{j
                 \in SI^*_{ID_i} \setminus SI_{ID_i}} a^{N+1-j} +
                 \sum_{j \in SI_{ID_i} \setminus SI^*_{ID_i}} a^{N+1-j} )
                 / \Delta T + r'_{i+1,0} } \\
            = &(g_1^b \cdot g_1^{\hat{\beta}_{ID^*_i}})^{\theta_{ID^*_i}}
           \Big( \prod_{j \in SI^*_{ID^*_i} \setminus SI_{ID_i}} g_1^{-a^{N+1-j}}
                 \prod_{j \in SI_{ID_i} \setminus SI^*_{ID^*_i}} g_1^{a^{N+1-j}}
           \Big)^{-h'_0 / \Delta T}
           H_1(T)^{r'_{i,0}},~ \db \\
    U_{i,2} = &g_1^{r_{i,0}}
             = \Big( \prod_{j \in SI^*_{ID^*_i} \setminus SI_{ID_i}} g_1^{-a^{N+1-j}}
                 \prod_{j \in SI_{ID_i} \setminus SI^*_{ID^*_i}} g_1^{a^{N+1-j}}
           \Big)^{-1 / \Delta T} g_1^{r'_{i,0}} \}_{0 \leq i \leq x}.
    \end{align*}

    \begin{align*}
    \{
    U_{i,1} = &\big( g_1^{\gamma_{ID^*_{i}}} \prod_{j \in SI_{ID_i}}
                 g_1^{\alpha^{N+1-j}} \big)^{\beta_{ID^*_i}} H_1(T)^{r_{i}}
            = \Big( g_1^{\theta_{ID_i}} \big( \prod_{j \in SI^*_{ID_i}}
                 g_1^{a^{N+1-j}} \big)^{-1} \prod_{j \in SI_{ID_i}}
                 g_1^{a^{N+1-j}} \Big)^{\beta_{ID^*_i}}
                 H_1(T)^{r_{i}} \}_{x < i \leq y}.
    \end{align*}
In case of $T = T^*$, we have that the update key is correctly distributed
from the setting $\{\beta_{ID^*_i} = b + \hat{\beta}_{ID^*_i}\}_{0 \leq i <
x}$ and $\{\gamma_{ID_i} = \theta_{ID_i} - \sum_{j \in SI^*_{ID_i}} a^{N+1-j}
\}_{0 \leq i \leq y}$ as the following equation
    \begin{align*}
    \{
    U_{i,1} &= \big( g_1^{\gamma_{ID_i}} \prod_{j \in SI_{ID_i}^*}
                     g_1^{\alpha^{N+1-j}} \big)^{\beta_{ID^*_i}}
                     \cdot H_1(T^*)^{r_{i}}
             = \Big( g_1^{\theta_{ID_i}} \big( \prod_{j \in SI^*_{ID_i}}
                     g_1^{a^{N+1-j}} \big)^{-1}
                     \cdot \prod_{j \in SI^*_{ID_i}} g_1^{a^{N+1-j}} \Big)
                     ^{b + \hat{\beta}_{ID_i}}
                     \cdot H_1(T^*)^{r_{i}} \\
            &= (g_1^b \cdot g_1^{\beta{ID_i}})^{\theta_{ID_i}} H_1(T^*)^{r_{i,0}}~
    \}_{1 \leq i \leq x-1},~\\
    \{
    U_{i,1} &= \big( g_1^{\gamma_{ID_i}} \prod_{j \in SI_{ID_i}^*}
                     g_1^{\alpha^{N+1-j}} \big)^{\beta_{ID^*_i}} \cdot H_1(T^*)^{r_{i}}
             = \Big( g_1^{\theta_{ID_i}} \big( \prod_{j \in SI^*_{ID_i}}
                     g_1^{a^{N+1-j}} \big)^{-1}
                     \cdot \prod_{j \in SI^*_{ID_i}} g_1^{a^{N+1-j}}
               \Big)^{\beta_{ID^*_i}}
                     \cdot H_1(T^*)^{r_{i}} \\
         &= (g_1^{\beta_{ID^*_i}})^{\theta_{ID_i}} H_1(T^*)^{r_{i}}\}_{x < i \leq y}.
    \end{align*}

Finally, we show that the distribution of the challenge ciphertext is
correct. If $Z = Z_0 = g_3^{a^{N+1}bc}$ is given, then the challenge
ciphertext is correctly distributed as the following equation
    \begin{align*}
    C &= \Omega^s \cdot M_{\delta}^*
      = \big( e( e(g_1^a, g_1^{a^N}), g_1^b) \cdot e( e(g_1^a, g_1^{a^N}),
              g_1^{\hat{\beta}_{ID_0}}) \big)^c \cdot M_{\delta}^* \\
      &= e( e(g_1^a, g_1^{a^N}), g_1^b)^c \cdot e( e(g_1^a, g_1^{a^N}),
              g_1^{\hat{\beta}_{ID_0}})^c \cdot M_{\delta}^* \\
      &= Z_0 \cdot  e( e(g_1^a, g_1^{a^N}), g_1^c )^{\hat{\beta}_{ID_0}}
         \cdot M_{\delta}^*
    \end{align*}
Otherwise, the component $C$ of the challenge ciphertext is independent of
$\delta$ in the $\mc{A}$'s view since $Z_1$ is a random element in $\G_3$.
This completes our proof.
\end{proof}

\section{Revocable HIBE with History-Free Updates} \label{sec:rhibe-hfu}

In this section, we first define the syntax and the security model of RHIBE
with history-free updates. Next, we propose another RHIBE scheme with short
private key and prove its security.

\subsection{Definition} \label{sec:rhibe-hfu-syntax}

\begin{definition}[Revocable HIBE: History-Free Update]
A revocable HIBE (RHIBE) scheme that is associated with the identity space
$\mc{I}$, the time space $\mc{T}$, and the message space $\mc{M}$, consists
of seven algorithms \tb{Setup}, \tb{GenKey}, \tb{UpdateKey}, \tb{DeriveKey},
\tb{Encrypt}, \tb{Decrypt}, and \tb{Revoke}, which are defined as follows:
\begin{description}
\item \tb{Setup}($1^\lambda, N, L$): The setup algorithm takes as input a
    security parameter $1^{\lambda}$, the maximum number $N$ of users in
    each depth, and the maximum depth $L$ of the identity. It outputs a
    master key $MK$, a revocation list $RL_{\epsilon}$, a state
    $ST_{\epsilon}$, and public parameters $PP$.

\item \tb{GenKey}($ID|_{\ell}, ST_{ID|_{\ell-1}}, PP$): The private key
    generation algorithm takes as input a hierarchical identity $ID|_{\ell}
    = (I_1, \ldots, I_{\ell}) \in \mc{I}^{\ell}$, a state
    $ST_{ID|_{\ell-1}}$, and public parameters $PP$. It outputs a private
    key $SK_{ID|_{\ell}}$ and updates the state $ST_{ID|_{\ell-1}}$.

\item \tb{UpdateKey}($T, RL_{ID|_{\ell-1}}, DK_{ID|_{\ell-1},T},
    ST_{ID|_{\ell-1}}, PP$): The update key generation algorithm takes as
    input update time $T \in \mc{T}$, a revocation list
    $RL_{ID|_{\ell-1}}$, a decryption key $DK_{ID|_{\ell-1},T}$, a state
    $ST_{ID|_{\ell-1}}$, and the public parameters $PP$. It outputs an
    update key $UK_{T,R|_{ID_{\ell-1}}}$ for $T$ and $R_{ID|_{\ell-1}}$
    where $R_{ID|_{\ell-1}}$ is the set of revoked identities at the time
    $T$.

\item \tb{DeriveKey}($SK_{ID|_{\ell}}, UK_{T,R_{ID|_{\ell-1}}}, PP$): The
    decryption key derivation algorithm takes as input a private key
    $SK_{ID|_{\ell}}$, an update key $UK_{T,R_{ID|_{\ell-1}}}$, and the
    public parameters $PP$. It outputs a decryption key $DK_{ID|_{\ell},T}$
    or $\perp$.

\item \tb{Encrypt}($ID|_{\ell}, T, M, PP$): The encryption algorithm takes
    as input a hierarchical identity $ID|_{\ell} = (I_1, \ldots, I_{\ell})
    \in \mc{I}$, time $T$, a message $M \in \mc{M}$, and the public
    parameters $PP$. It outputs a ciphertext $CT_{ID|_{\ell},T}$ for
    $ID|_{\ell}$ and $T$.

\item \tb{Decrypt}($CT_{ID|_{\ell},T}, DK_{ID'|_{\ell},T'}, PP$): The
    decryption algorithm takes as input a ciphertext $CT_{ID|_{\ell},T}$, a
    decryption key $DK_{ID'|_{\ell},T'}$, and the public parameters $PP$.
    It outputs an encrypted message $M$ or $\perp$.

\item \tb{Revoke}($ID|_{\ell}, T, RL_{ID|_{\ell-1}}, ST_{ID|_{\ell-1}}$):
    The revocation algorithm takes as input a hierarchical identity
    $ID|_{\ell}$ and revocation time $T$, a revocation list
    $RL_{ID|_{\ell-1}}$, and a state $ST_{ID|_{\ell-1}}$. It updates the
    revocation list $RL_{ID|_{\ell-1}}$.
\end{description}
The correctness property of RHIBE is defined as follows: For all $MK$, $PP$
generated by $\tb{Setup} (1^{\lambda}, N, L)$, $SK_{ID_{\ell}}$ generated by
$\tb{GenKey} (ID|_{\ell}, ST_{ID|_{\ell-1}}, PP)$ for any $ID|_{\ell}$,
$UK_{T,R_{ID|_{\ell-1}}}$ generated by $\tb{UpdateKey}$ $(T,
RL_{ID|_{\ell-1}}, \lb DK_{ID|_{\ell-1},T}, ST_{ID|_{\ell-1}}, PP)$ for any
$T$ and $RL_{ID|_{\ell-1}}$, $CT_{ID|_{\ell},T}$ generated by $\tb{Encrypt}
(ID|_{\ell}, T, M, PP)$ for any $ID|_{\ell}$, $T$, and $M$, it is required
that
\begin{itemize}
\item If $(ID|_{\ell} \notin R_{ID|_{\ell-1}})$, then $\tb{DeriveKey}
    (SK_{ID|_{\ell}}, UK_{T,R_{ID|_{\ell-1}}}, PP) = DK_{ID|_{\ell},T}$.

\item If $(ID|_{\ell} \in R_{ID|_{\ell-1}})$, then $\tb{DeriveKey}
    (SK_{ID|_{\ell}}, UK_{T,R_{ID|_{\ell-1}}}, PP) = \perp$ with all but
    negligible probability.

\item If $(ID|_{\ell} = ID'|_{\ell}) \wedge (T = T')$, then $\tb{Decrypt}
    (CT_{ID|_{\ell},T}, DK_{ID'|_{\ell},T'}, PP) = M$.

\item If $(ID|_{\ell} \neq ID'|_{\ell}) \vee (T \neq T')$, then
    $\tb{Decrypt} (CT_{ID|_{\ell},T}, DK_{ID'|_{\ell},T'}, PP) = \perp$
    with all but negligible probability.
\end{itemize}
\end{definition}

\subsection{Construction} \label{sec:rhibe-hfu-scheme}

Our RHIBE scheme from three-leveled multilinear maps is described as follows:

\begin{description}
\item [\tb{RHIBE.Setup}($1^\lambda, N, L$):] Let $N$ be the maximum number
    users in each depth and $L$ be the maximum depth of the hierarchical
    identity.
    \begin{enumerate}
    \item It first generates a multilinear group $\vec{\G} = (\G_1, \G_2,
        \G_3)$ of prime order $p$. Let $GDS_{MLM} = (p, \vec{\G}, \{
        e_{1,1}, e_{1,2}, e_{2,1} \}, g_1, g_2, g_3)$ be the description
        of the multilinear group where $g_1, g_2, g_3$ are generators of
        $\G_1, \G_2, \G_3$ respectively.

    \item It obtains $MK_{HIBE}, PP_{HIBE}$ by running $\tb{HIBE.Setup}
        (GDS_{MLM}, N, L)$. It also obtains $MK_{BE} = (\alpha, \gamma),
        PP_{BE}$ by running $\tb{PKBE.Setup}(GDS_{MLM}, N)$.

    \item It selects a random exponent $\beta_{\epsilon} \in \Z_p$ and
        saves $(\beta_{\epsilon}, \gamma_{\epsilon})$ to $ST_{\epsilon}$
        where $\beta_{\epsilon} = \beta_{ID|_0}$ and $\gamma_{\epsilon} =
        \gamma_{ID_0} = \gamma$. It outputs a master key $MK = \alpha$,
        an empty revocation list $RL_{\epsilon}$, a state
        $ST_{\epsilon}$, and public parameters
        \begin{align*}
        PP &= \Big(
           GDS_{MLM},~ PP_{HIBE},~ PP_{BE},~
           g_2^{\alpha^{N+1}},~ g_2^{\beta_{\epsilon}},~
           \Omega = g_3^{\alpha^{N+1} \beta_{\epsilon}}
        \Big).
        \end{align*}
    \end{enumerate}

\item [\tb{RHIBE.GenKey}($ID|_{\ell}, ST_{ID|_{\ell-1}}, PP$):] Let
    $ID|_{\ell} = (I_1, \ldots, I_{\ell}) \in  \mc{I}^{\ell}$. It proceeds
    as follows:
    \begin{enumerate}
    \item If a tuple $(\beta_{ID|_{\ell-1}}, \gamma_{ID|_{\ell-1}})$
        exist in $ST_{ID|_{\ell-1}}$, then it retrieves
        $(\beta_{ID|_{\ell-1}}, \gamma_{ID|_{\ell-1}})$ from
        $ST_{ID|_{\ell-1}}$. Otherwise, it selects random exponents
        $\beta_{ID|_{\ell-1}}, \gamma_{ID|_{\ell-1}} \in \Z_p$ and saves
        $(\beta_{ID|_{\ell-1}}, \gamma_{ID|_{\ell-1}})$ to
        $ST_{ID|_{\ell-1}}$.

    \item It assigns a unique index $d_{\ell} \in \mc{N}$ to the identity
        $ID|_\ell$ and adds a tuple $(ID|_\ell, d_\ell)$ to
        $ST_{ID|_{\ell-1}}$. It obtains a private key $SK_{BE,d_\ell} =
        K_{BE}$ by running $\tb{PKBE.GenKey} (d_\ell,
        \gamma_{ID|_{\ell-1}}, PP_{BE})$.
        Next, it selects a random exponent $r_{\ell,1} \in \Z_p$ and
        creates a level private key
        \begin{align*}
        LSK_{\ell} = \Big(
            K_{\ell,0} = K_{BE} \cdot F_{1,\ell}(I_{\ell})^{-r_{\ell,1}},~
            K_{\ell,1} = g_1^{-r_{\ell,1}}
        \Big) \in \G_1^2.
        \end{align*}

    \item Finally, it outputs a private key $SK_{ID|_\ell} = \big( \{
        d_\ell, LSK_\ell \} \big)$.
    \end{enumerate}

\item [\tb{RHIBE.UpdateKey}($T, RL_{ID|_{\ell-1}}, DK_{ID|_{\ell-1}, T},
    ST_{ID|_{\ell-1}}, PP$):] Let $DK_{ID|_{\ell-1},T} = (D_0, \{ D_i
    \}_{i=1}^{\ell-1}, D_{L+1})$ where $\ell \geq 1$. It proceeds as
    follows:
    \begin{enumerate}
    \item It defines a revoked set $R_{ID|_{\ell-1}}$ of user identities
        at time $T$ from $RL_{ID|_{\ell-1}}$. From $R_{ID|_{\ell-1}}$, it
        defines a revoked index set $RI_{ID|_{\ell-1}} \subseteq \mc{N}$
        by using $ST_{ID|_{\ell-1}}$ since $ST_{ID|_{\ell-1}}$ contains
        $(ID|_\ell, d_{\ell})$. After that, it defines a non-revoked
        index set $SI_{ID|_{\ell-1}} = \mc{N} \setminus
        RI_{ID|_{\ell-1}}$.

    \item It retrieves $(\beta_{ID|_{\ell-1}}, \gamma_{ID|_{\ell-1}})$
        from $ST_{ID|_{\ell-1}}$. It obtains $CH_{BE} = (E_0, E_1)$ by
        running $\tb{PKBE.Encrypt} \lb (SI_{ID|_{\ell-1}},
        \beta_{ID|_{\ell-1}}, Y_{ID|_{\ell-1}} =
        g_1^{\gamma_{ID|_{\ell-1}}}, PP_{BE})$. Next, it selects a random
        exponent $r_{\ell-1} \in \Z_p$ and creates a level update key
        \begin{align*}
        LUK_{\ell-1} = \Big(
            U_{\ell-1,0} = E_0,~
            U_{\ell-1,1} = E_1 \cdot H_1(T)^{r_{\ell-1}},~
            U_{\ell-1,2} = g_1^{-r_{\ell-1}}
        \Big) \in \G_1^3.
        \end{align*}

    \item If $\ell = 1$, then $DK_{ID|_{\ell-1},T} = MK$ and creates a
        partial decryption key $PDK_0 = \big( P_0 = 1_{\G_2},~ P_{L+1} =
        1_{\G_2} \big)$. Otherwise ($\ell \geq 2$), then it creates a
        partial decryption key
        \begin{align*}
        PDK_{\ell-1} = \Big(
            P_0 = D_0 \cdot \big( g_2^{\alpha^{N+1}} \big)^{-\beta_{ID|_{\ell-1}}},~
            \big\{ P_i = D_i \big\}_{i=1}^{\ell-1},~
            P_{L+1} = D_{L+1}
        \Big) \in \G_2^{\ell+2}.
        \end{align*}

    \item Finally, it outputs an update key $UK_{T,R_{ID|_{\ell-1}}} =
        \big( PDK_{\ell-1}, \{ SI_{ID|_{\ell-1}}, LUK_{\ell-1} \} \big)$.
   \end{enumerate}

\item [\tb{RHIBE.DeriveKey}($SK_{ID|_{\ell}}, UK_{T,R_{ID|_{\ell-1}}},
    PP$):] Let $SK_{ID|_\ell} = ( \{ d_\ell, LSK_\ell \} )$ where $LSK_\ell
    = ( K_0, K_1 )$ and $\ell \geq 1$, and $UK_{T,R_{ID|_{\ell-1}}} = (
    PDK_{\ell-1}, \{ SI_{ID|_{\ell-1}}, LUK_{\ell-1} \} )$ where
    $PDK_{\ell-1} = ( P_0, \{ P_i \}_{i=1}^{\ell-1}, P_{L+1} )$ and
    $LUK_{\ell-1} = ( U_0, U_1, U_2 )$.
    If $ID|_{\ell} \in R_{ID|_{\ell-1}}$, then it outputs $\perp$ since the
    identity $ID_\ell$ is revoked. Otherwise, it proceeds the following
    steps:
    \begin{enumerate}
    \item For $i = \ell$, it retrieves $\{ d_i, LSK_i = ( K_{i,0},
        K_{i,1} ) \}$ and $\{ SI_{ID|_{i-1}}, LUK_{i-1} = ( U_{i-1,0},
        U_{i-1,1}, U_{i-1,2} ) \}$ and computes the following components
        \begin{align*}
        A_{i,0} &= e_{1,1} (X_{d_i}, U_{i,1}) \cdot
                   e_{1,1} \big( U_{i-1,0}, K_{i,0}
                   \prod_{j \in SI_{ID|_{i-1}}, j \neq d_i} X_{N+1 -j + d_i}
                   \big)^{-1},~ \\
        A_{i,1} &= e_{1,1} (U_{i-1,0}, K_{i,1}),~
        A_{i,2}  = e_{1,1} (X_{d_i}, U_{i-1,2}).
        \end{align*}

    \item Next, it derives a temporal decryption key
        \begin{align*}
        TDK_{ID|_{\ell},T} = \Big(
            D_0 = P_0 \cdot A_{\ell,0},~
            \{ D_i = P_i \}_{i=1}^{\ell-1},~
            D_{\ell} = A_{\ell,1},~
            D_{L+1} = P_{L+1} \cdot A_{\ell,2}
        \Big) \in \G_2^{\ell+2}.
        \end{align*}

\item Finally, it outputs a decryption key $DK_{ID|_{\ell},T}$ by running
    $\tb{HIBE.RandKey} (TDK_{ID|_{\ell},T}, \lb PP_{HIBE})$.
    \end{enumerate}

\item [\tb{RHIBE.Encrypt}($ID|_{\ell}, T, M, PP$):] Let $ID|_{\ell} = (I_1,
    \ldots, I_{\ell})$. It first chooses a random exponent $s \in \Z_p$ and
    obtains $CH_{HIBE}$ by running $\tb{HIBE.Encrypt} (ID|_{\ell}, T, s,
    PP_{HIBE})$. It outputs a ciphertext $CT_{ID|_\ell,T} = \big( C =
    \Omega^s \cdot M,~ CH_{HIBE} \big)$.

\item [\tb{RHIBE.Decrypt}($CT_{ID,T}, DK_{ID',T'}, PP$):] Let $CT_{ID,T} =
    (C, CH_{HIBE})$. If $(ID = ID') \wedge (T = T')$, then it obtains
    $EK_{HIBE}$ by running $\tb{HIBE.Decrypt} (CH_{HIBE}, DK_{ID', T'},
    PP_{HIBE})$ and outputs the message $M$ by computing $M = C \cdot
    EK^{-1}_{HIBE}$. Otherwise, it outputs $\perp$.

\item [\tb{RHIBE.Revoke}($ID|_{\ell}, T, RL_{ID|_{\ell-1}},
    ST_{ID|_{\ell-1}}$):] If $(ID|_\ell, -) \notin ST_{ID|_{\ell-1}}$, then
    it outputs $\perp$ since the private key of $ID|_{\ell}$ was not
    generated. Otherwise, it updates $RL_{ID|_{\ell-1}}$ by adding
    $(ID|_{\ell}, T)$ to $RL_{ID|_{\ell-1}}$.
\end{description}

\subsection{Correctness}

Let $SK_{ID_{\ell}}$ be a private key for an identity $ID_{\ell}$ that is
associated with an index $d_{ID_{\ell-1}}$, and $UK_{T,R_{ID_{\ell-1}}}$ be
an update key for a time $T$ and a revoked identity set $R_{ID_{\ell-1}}$.
We have
    \begin{align*}
    A_{i,0}
    &=  EK_{BE} \cdot e_{1,1} (X_{d_i}, H_1(T)^{r_i}) \cdot
        e_{1,1} (E_0, F_{1,i}(I_i)^{r_{i,1}}),~ \\
    &=  g_2^{\alpha^{N+1} \beta_{ID_{i-1}}} \cdot
        H_2(T)^{\alpha^{d_i} r_i} \cdot
        F_{2,i}(I_i)^{\beta_{ID_{i-1}} r_{i,1}},~ \db \\
    A_{i,1}
    &=  e_{1,1} (U_{i-1,0}, K_{i,1})
     =  e_{1,1} \big( g_1^{\beta_{ID_{i-1}}}, g_1^{-r_{i,1}} \big)
     =  g_2^{-\beta_{ID_{i-1}} r_{i,1}},~ \\
    A_{i,2}
    &=  e_{1,1} (X_{d_i}, U_{i-1,2})
     =  e_{1,1} \big( g_1^{\alpha^{d_i}}, g_1^{-r_i} \big)
     =  g_2^{-\alpha^{d_i} r_i}.
    \end{align*}

If $ID_{\ell} \notin R_{ID_{\ell-1}}$, then the decryption key derivation
algorithm first correctly derives temporal decryption key as
    \begin{align*}
    D_0
    &=  P_0 \cdot A_{\ell,0}
     =  \prod_{i=1}^{\ell-1} g_2^{-\alpha^{N+1} \beta_{ID_i}} \cdot
        \prod_{i=1}^{\ell} A_{i,0} \\
    &=  \prod_{i=1}^{\ell-1} g_2^{-\alpha^{N+1} \beta_{ID_i}} \cdot
        \prod_{i=1}^{\ell}
        g_2^{\alpha^{N+1} \beta_{ID_{i-1}}}
        H_2(T)^{\alpha^{d_i} r_i} F_{2,i}(I_i)^{\beta_{ID_{i-1}} r_{i,1}} \\
    &=  g_2^{\alpha^{N+1}\beta_{\epsilon}} \cdot
        \prod_{i=1}^{\ell} F_{2,i}(I_{i})^{\beta_{ID_{i-1}} r_{i,1}} \cdot
        H_1(T)^{\sum_{i=1}^{\ell} \alpha^{d_i} r_i},~ \db \\
    D_{i}
    &=  g_2^{-\beta_{ID_{i-1}} r_{i,1}}~~~ \forall i \in [\ell],~
    D_{L+1}
     =  P_{L+1} \cdot A_{\ell,2}
     =  \prod_{i=1}^{\ell} A_{i,2}
     =  g_2^{-\sum_{i=1}^{\ell} \alpha^{d_i} r_i}.
    \end{align*}

\subsection{Security Analysis}

To prove the security of our RHIBE scheme via history-free approach, we
carefully combine the partitioning methods of the BGW-PKBE scheme
\cite{BonehGW05}, the BB-HIBE scheme \cite{BonehB04e}, and our cancelation
technique.

\begin{theorem} \label{thm:rhibe-hfu-srlind}
The above RHIBE scheme is SRL-IND secure if the $(3,N)$-MDHE assumption holds
where $N$ is the maximum child number of users in the system.
\end{theorem}

\begin{proof}
Suppose there exists an adversary $\mc{A}$ that attacks the above RHIBE scheme
with a non-negligible advantage. A meta-simulator $\mc{B}$ that solves the MDHE
assumption using $\mc{A}$ is given: a challenge tuple
    $D = \big( (p,\G_1, \G_2, \G_3), g_1, g_1^{a}, g_1^{a^2}, \ldots,
    g_1^{a^N}, g_1^{a^{N+2}}, \ldots, g_1^{a^{2N}}, g_1^b, g_1^c \big)$
    and $Z$
where $Z = Z_0 = g_3^{a^{N+1} bc}$ or $Z = Z_1 \in \G_3$.
Note that a challenge tuple
    $D_{BDHE} = \big( (p,\G_1, \G_2), g_1, g_1^{a}, g_1^{a^2}, \ldots,
    g_1^{a^N}, g_1^{a^{N+2}}, \ldots, g_1^{a^{2N}}, g_1^b \big) $
for the BDHE assumption can be derived from the challenge tuple $D$ of the
MDHE assumption. Let $\mc{B}_{HIBE}$be the simulator in the security proof
of Theorem \ref{thm:bb-hibe-indcpa} and $\mc{B}_{PKBE}$ be a simulator in security
proof of Theorem \ref{thm:bgw-pkbe-indcpa}.
Then $\mc{B}$ that interacts with $\mc{A}$ is described as follows:

\vs \noindent \textbf{Init:} $\mc{A}$ initially submits a challenge identity
$ID^*_{\ell^*} = (I^*_1,\ldots,I^*_{\ell^*})$, a challenge time $T^*$, and a
revoked identity set $R^* = (R^*_{ID^*_0},\ldots,R^*_{ID^*_{\ell^*-1}})$ at
the time $T^*$.
It first sets a state $ST$ and a revocation list $RL$ as empty one. For each
$ID \in \{ ID^* \} \cup R^*$, it selects an index $d_{ID^*_{x}} \in \mc{N}_x$
such that $(-, d_{ID^*_{x}}) \notin ST_{ID^*_{x}}$ and adds $(ID,
d_{ID^*_{x}})$ to $ST_{ID^*_{x}}$. Let $RI^* =
(RI^*_{ID_0},\ldots,RI^*_{ID_{\ell^*-1}}) \subseteq \mc{N}$ be the revoked
index set of $R^*$ at the time $T^*$ and $SI^* = (SI^*_{ID_1},\ldots,
SI^*_{ID_{\ell^*-1}}$) be the non-revoked index set at the time $T^*$ such
that $SI^*_{ID_x} = \mc{N}_x \setminus RI^*_{ID_x}$.

\svs \noindent \textbf{Setup:} $\mc{B}$ submits $ID^*_{\ell^*}$ and $T^*$ to
$\mc{B}_{HIBE}$ and receives $PP_{HIBE}$. It also submits $SI^*_{ID^*_{0}}$
to $\mc{B}_{PKBE}$ and receives $PP_{PKBE}$.
$\mc{B}$ first chooses random exponents $\theta_{ID^*_0},\ldots, $
$\theta_{ID^*_{\ell^*-1}} \in \Z_p$.
It implicitly sets $\alpha = a, \beta_{\epsilon} = \beta_{ID_0} = b,
\{\beta_{ID_i} = b + \hat{\beta}_{ID_i} \}_{1 \leq i \leq \ell^*-1},
\{\gamma_{ID^*_x} = \theta_{ID^*_x} - \sum_{j \in SI_{ID^*_x}} a^{N+1-j}\}_
{0 \leq x \leq \ell^*-1}$ and publishes the public parameters $PP$ as
    \begin{align*}
    PP = \big(
        PP_{HIBE},~ PP_{PKBE},~
       \Omega = e \big( e (g_1^{\alpha}, g_1^{\alpha^N}),
            g_1^b \big)
        = g_3^{\alpha^{N+1} \beta_{\epsilon}}
        \big).
    \end{align*}

\vs \noindent \textbf{Phase 1:} $\mc{A}$ adaptively requests a polynomial
number of private key, update key, and decryption key queries.

\svs \noindent If this is a private key query for an identity $ID_{\ell} =
(I_1,\ldots,I_{\ell})$, then $\mc{B}$ proceeds as follows:
Note that $SK_{ID_{\ell}} = (\{d_{\ell}, LSK_{\ell}\})$ where $LSK_{\ell} =
(K_{\ell,0}, K_{\ell,1})$.

\begin{itemize}
\item \textbf{Case} $ID_{\ell-1} \notin \textbf{Prefix}(ID^*_{\ell^*})$:
    In this case, it can normally generate the state $ST_{ID_{\ell-1}}$
    by himself.
    It first normally sets $ST_{ID_{\ell-1}}$ where the index $d_{ID_{\ell-1}}$ is
    associated with $ID_{\ell}$. And, it selects a random exponent
    $\gamma_{ID_{\ell}}$. It obtains $SK_{ID_{\ell}}$ by running \textbf
    {RHIBE.GenKey}$(ID_{\ell},$ $ST_{ID_{\ell-1}}, PP)$.

\item \textbf{Case} $ID_{\ell-1} \in \textbf{Prefix}(ID^*_{\ell^*})$:
    In this case, it first retrieves a tuple $(ID_{\ell}, d_{i})$
    from $ST_{ID^*_{\ell-1}}$ where the index $d_{i}$ is associated
    with $ID_{\ell}$.
    \begin{itemize}
    \item \textbf{Case} $ID_{\ell} \in R^*_{ID^*_{\ell-1}}$:
    In this case, the simulator can use the partitioning method of Boenh et al.
    \cite{BonehGW05}.
    Next, it
    selects a random exponent $r_{\ell} \in \Z_p$ and creates a private key
    $SK_{ID_{\ell}}$ as
    \begin{align*}
    LSK_{\ell} = \Big(
    K_{\ell, 0} &= (g_1^{a^{d_i}})^{\theta}(\prod_{j\in SI^*} g_1^{a^{N+1-j+d_i}})^{-1}
                   \cdot F_{1,\ell} (I_{\ell})^{-r_{\ell}},~
    K_{\ell, 1}  = g_1^{-r_{\ell}} \Big).
    \end{align*}

\item \textbf{Case} $ID_{\ell} \notin R^*_{ID^*_{\ell-1}}$: In this case, we
    have $I_{\ell} \neq I^*_{\ell}$ from the restriction of Definition
    \ref{def:rhibe-hpu-srlind} and the simulator can use the paritioning method
    of Boneh and Boyen \cite{BonehB04e}.
    It first selects an index $d_{i} \in \mc{N}_{\ell}$ such that
    $(-, d_{i}) \notin ST_{ID^*_{\ell-1}}$ and adds $(ID_{\ell},
    d_{i})$ to $ST_{ID^*_{\ell-1}}$.
    Next, it selects a random exponent $r'_{\ell} \in \Z_p$ and
    creates a private key $SK_{ID}$ by implicitly setting $r_{\ell} = -a
    / \Delta ID + r'_{\ell}$ as
    \begin{align*}
    LSK_{\ell} = \Big(
    K_{\ell,0} &=  g_1^{a^{d_{i}} \theta_{ID^*_{\ell-1}}} \prod_{j \in
                SI^*_{ID^*_{\ell-1}} \setminus \{ d_{i} \} }
            g_1^{-a^{N+1-j+d_{i}}} (g_1^a)^{f'_0 / \Delta ID}
                F_{1, \ell}(I_{\ell})^{r'_{\ell}},~\\
    K_{\ell,1}  &=  (g_1^a)^{-1/\Delta ID} g_1^{-r'_{\ell}} \Big).
    \end{align*}
    \end{itemize}
\end{itemize}

\noindent If this is an update key query for an identity $ID_{\ell-1} =
(I_1,\ldots,I_{\ell-1})$ and a time $T$, then $\mc{B}$ defines a revoked
identity set $R_{ID_{\ell-1}}$ at the time $T$ from $RL_{ID_{\ell-1}}$ and
proceeds as follows: Note that $UK_{T, R_{ID_{\ell-1}}} = (PDK_{\ell-1},
\{SI_{ID_{\ell-1}}, LUK_{\ell-1}\})$ where $LUK_{\ell-1} = (U_{\ell-1,0},
U_{\ell-1,1}, U_{\ell-1,2})$ and $PDK_{\ell-1} = (P_0,
\{P_i\}_{i=1}^{\ell-1}), P_{L+1}$.

\begin{itemize}
\item \textbf{Case} $T \neq T^*$: It first sets a revoked index set
    $RI_{ID_{\ell-1}}$ of $R_{ID_{\ell-1}}$ by using $ST_{ID_{\ell-1}}$.
    It also sets $SI_{ID_{\ell-1}} = \mc{N}_1 \setminus RI_{ID_{\ell-1}}$.
    And, it also sets $SI_{ID_{\ell-1}} = \mc{N}_1 \setminus
    RI_{ID_{\ell-1}}$.
    \begin{itemize}
    \item \textbf{Case} $ID_{\ell-1} = ID_0$ : In this case, the simultor can
    use the partitioning method of Boneh and Boyen \cite{BonehB04e}. It selects
    a random exponent $r'_0 \in \Z_p$ and creates an update key $UK_{T,R}$ by
    implicitly setting $r_0 = - (-\sum_{j \in SI^*_{ID_0}
    \setminus SI_{ID_0}} a^{N+1-j} + \sum_{j \in SI_{ID_0} \setminus SI^*_{ID_0}}
    a^{N+1-j}) / \Delta T + r'_0$ as
    \begin{align*}
    LUK_{\ell-1} = \Big(
    &U_{\ell-1, 0} = g_1^b,~ \\
    &U_{\ell-1, 1} = (g_1^b)^{\theta_{ID_0}}
           \Big( \prod_{j \in SI^*_{ID_0} \setminus SI_{ID_0}}
               g_1^{-a^{N+1-j}}
                 \prod_{j \in SI_{ID_0} \setminus SI^*_{ID_0}}
               g_1^{a^{N+1-j}}
           \Big)^{-h'_0 / \Delta T}
           H_1(T)^{r'_0},~\\
    &U_{\ell-1, 2}  = \Big( \prod_{j \in SI^*_{ID_0} \setminus SI_{ID_0}}
               g_1^{-a^{N+1-j}}
                 \prod_{j \in SI_{ID_0} \setminus SI^*_{ID_0}}
               g_1^{a^{N+1-j}}
           \Big)^{-1 / \Delta T} g_1^{r'_0} \Big),\\
    PDK_{\ell-1} = \Big(
            &P_0 = 1_{\G_2},~
             P_{L+1} = 1_{\G_2} \Big).
    \end{align*}
    \item \textbf{Case} $ID_{\ell-1} \in$ \textbf{Prefix}$(ID^*_{\ell^*})$ :
    In this case, the simulator can use the partitioning method of Boneh and Boyen
    \cite{BonehB04e} to create a level update key, and the cancelation technique
    by using the session key of PKBE to create a partial decryption key. It selects
    a random exponent $r'_0 \in \Z_p$ and creates an update key $UK_{T,R}$ by
    implicitly setting $r_0 = - (-\sum_{j \in SI^*_{ID_0}
    \setminus SI_{ID_0}} a^{N+1-j} + \sum_{j \in SI_{ID_0} \setminus SI^*_{ID_0}}
    a^{N+1-j}) / \Delta T + r'_0$ as
    \begin{align*}
    LUK_{\ell-1} = \Big(
    U_{\ell-1, 0} = &g_1^{b + \hat{\beta}_{ID_{\ell-1}}},~ \\
    U_{\ell-1, 1} = &(g_1^{b + \hat{\beta}_{ID_{\ell-1}}})^{\theta_{ID_{\ell-1}}}
           \Big( \prod_{j \in SI^*_{ID^*_{\ell-1}} \setminus SI_{ID^*_{\ell-1}}}
                 g_1^{-a^{N+1-j}}
           \prod_{j \in SI_{ID^*_{\ell-1}} \setminus SI^*_{ID^*_{\ell-1}}}
           g_1^{a^{N+1-j}}
                \Big)^{\hat{\beta}_{ID_{\ell-1}}}\\
           &\times
           \Big( \prod_{j \in SI^*_{ID_{\ell-1}} \setminus SI_{ID_{\ell-1}}}
               g_1^{-a^{N+1-j}}
                 \prod_{j \in SI_{ID_{\ell-1}} \setminus SI^*_{ID_{\ell-1}}}
               g_1^{a^{N+1-j}}
           \Big)^{-h'_0 / \Delta T}
           H_1(T)^{r'_0},~\\
    U_{\ell-1, 2}  = &\Big( \prod_{j \in SI^*_{ID_{\ell-1}} \setminus SI_{ID_{\ell-1}}}
               g_1^{-a^{N+1-j}}
                 \prod_{j \in SI_{ID_{\ell-1}} \setminus SI^*_{ID_{\ell-1}}}
               g_1^{a^{N+1-j}}
           \Big)^{-1 / \Delta T} g_1^{r'_0} \Big),\\
    PDK_{\ell-1} = \Big(
            P_0 = &\prod_{i=1}^{\ell-1}F_{2,i}(I_i)^{r_{i}} \cdot
               H(T)^{r_{0}} \cdot e (g_1^a, g_1^{a^N})
               ^{-\hat{\beta}_{ID^*_{\ell-1}}},~
      \big\{ P_i = g_2^{r_{i}} \big\}_{i=1}^{\ell-1},~
            P_{L+1} = g_2^{r_{L+1}} \Big).
    \end{align*}
    \item \textbf{Case} $ID_{\ell-1} \notin$
        \textbf{Prefix}$(ID^*_{\ell^*})$ : In this case, the simulator
        can obtain the decryption key $DK_{ID_{\ell-1}, T} =
        (D_0,\ldots,D_{\ell-1}, D_{L+1})$ by requesting an RHIBE
        decryption key query. Next, it can nomally create an update key
        by running \textbf{RHIBE.}\textbf{UpdateKey} ($T,
        RL_{ID_{\ell-1}}, DK_{ID_{\ell-1}, T}, ST_{ID_{\ell-1}}, PP$).
    \end{itemize}

\item \textbf{Case} $T = T^*$: For each $ID \in R^*_{ID^*_{\ell-1}}$, it adds
    $(ID, T^*)$ to $RL_{ID^*_{\ell-1}}$ if $(ID, T') \notin RL_{ID^*_{\ell-1}}$
    for any $T' \leq T^*$.
    \begin{itemize}
    \item \textbf{Case} $ID_{\ell-1} = ID_0$ : In this case, the simulator can
    use the partitioning method of Boneh et al. \cite{BonehGW05}. It selects
    random exponent $r_{0} \in \Z_p$ and creates an update key $UK_{T,
    R_{ID_{\ell-1}}}$ as
    \begin{align*}
    LUK_{\ell-1} &= \Big(
    U_{\ell-1, 0} = g_1^b,~
    U_{\ell-1, 1}  = (g_1^b)^{\theta_{ID^*_{\ell-1}}} \cdot H_1(T^*)^{r_0},~
    U_{\ell-1, 2}  = g_1^{-r_0} \Big),~\\
    PDK_{\ell-1} &= \Big(
             P_0 = 1_{\G_2},~
             P_{L+1} = 1_{\G_2} \Big).
    \end{align*}
    \item \textbf{Case} $ID_{\ell-1} \in$
        \textbf{Prefix}$(ID^*_{\ell^*})$ : In this case, the simulator
        can use the partitioning method of Boneh et al. \cite{BonehGW05}.
        It selects random exponents $r_{0},\ldots,r_{\ell-1}, r_{L+1} \in
        \Z_p$ and creates an update key $UK_{T,R}$ as
    \begin{align*}
    LUK_{\ell-1} = \Big(
    &U_{\ell-1, 0} = g_1^b \cdot g_1^{\hat{\beta}_{ID^*_{\ell-1}}},~
     U_{\ell-1, 1}  = (g_1^b \cdot g_1^{\hat{\beta}_{ID^*_{\ell-1}}})^{\theta
               _{ID^*_{\ell-1}}}
               H_1(T^*)^{r_0},~
     U_{\ell-1, 2}  = g_1^{-r_0} \Big),~\\
    PDK_{\ell-1} = \Big(
             &P_0 = \prod_{i=1}^{\ell-1}F_{2,i}(I_i)^{r_{i}} \cdot
               H(T)^{r_{0}} \cdot e (g_1^a, g_1^{a^N})
               ^{-\hat{\beta}_{ID^*_{\ell-1}}},~
       \big\{ P_i = g_2^{r_{i}} \big\}_{i=1}^{\ell-1},~
          P_{L+1} = g_2^{r_{L+1}} \Big).
    \end{align*}
    \item \textbf{Case} $ID_{\ell-1} \notin$
        \textbf{Prefix}$(ID^*_{\ell^*})$ : In this case, the simulator
        can obtain the decryption key $DK_{ID_{\ell-1}, T} =
        (D_0,\ldots,D_{\ell-1}, D_{L+1})$ by requesting an RHIBE
        decryption key query. Next, it can nomally create an update key
        by running \textbf{RHIBE.}\textbf{UpdateKey} ($T,
        RL_{ID_{\ell-1}}, DK_{ID_{\ell-1}, T}, ST_{ID_{\ell-1}}, PP$).
    \end{itemize}

\end{itemize}

\noindent If this is a decryption key query for an identity $ID =(I_1,\ldots,
I_{\ell})$  and a time $T$, then $\mc{B}$ proceeds as follows:
It requests an HIBE private key for $ID$ and $T$ to $\mc{B}_{HIBE}$ and
receives $SK_{HIBE, ID, T}$. Next, it sets the decryption key $DK_{ID, T}
= SK_{HIBE, ID, T}$.\\

\noindent \textbf{Challenge}: $\mc{A}$ submits two challenge messages
$M_0^*, M_1^*$. $\mc{B}$ chooses a random bit $\delta \in \bits$ and
proceed as follows:
It requests the challenge ciphertext for $ID^*$ and $T^*$ to $\mc{B}_{HIBE}$
and receives $CH_{HIBE, ID^*, T^*}$. Next, it sets the challenge ciphertext
$CT_{ID^*, T^*} = (Z \cdot M_{\delta}^*, CH_{HIBE, ID^*, T^*})$.
\\

\noindent \textbf{Phase 2}: Same as Phase 1.

\svs \noindent \textbf{Guess}: Finally, $\mc{A}$ outputs a guess $\delta' \in
\bits$. $\mc{B}$ outputs $0$ if $\delta = \delta'$ or $1$ otherwise.
\end{proof}

\begin{lemma} \label{lem:ribe-basic-dist}
The distribution of the above simulation is correct if $Z = Z_0$, and the
challenge ciphertext is independent of $\delta$ in the adversary's view if $Z
= Z_1$.
\end{lemma}

\begin{proof}
We show that the distribution of private keys is correct. In case of
$ID_{\ell} \in R^*_{ID^*_{\ell-1}}$ and $ID_{\ell-1} \in \textbf{Prefix}
(ID^*_{\ell^*})$, we have that the private key is correctly distributed
from the setting $\gamma_{ID^*_{\ell-1}} = \theta_{ID^*_{\ell-1}} - \sum_{j
\in SI^*_{ID^*_{\ell-1}}} a^{N+1-j}$ as the following equation
    \begin{align*}
    K_{\ell, 0} = g_1^{a^{d_{\ell}}\gamma_{ID^*_{\ell-1}}}
                   \cdot F_{1,\ell} (I_{\ell})^{-r_{\ell}}
                = g_1^{a^{d_{\ell}}(\theta_{ID^*_{\ell-1}} - \sum_{j
                                   \in SI^*_{ID^*_{\ell-1}}} a^{N+1-j})}
                = (g_1^{a^{d_{\ell}}})^{\theta}(\prod_{j\in SI^*}
                   g_1^{a^{N+1-j+d_{\ell}}})^{-1}
                   \cdot F_{1,\ell} (I_{\ell})^{-r_{\ell}}.
    \end{align*}
In case of $ID_{\ell} \notin R^*_{ID^*_{\ell-1}}$ and $ID_{\ell-1} \in \textbf{Prefix}
(ID^*_{\ell^*})$, we have that the private key is correctly distributed
from the setting $\gamma_{ID^*_{\ell-1}} = \theta_{ID^*_{\ell-1}} - \sum_{j
\in SI^*_{ID^*_{\ell-1}}} a^{N+1-j}$ and $r_{\ell} = -a / \Delta I_{\ell} +
r'_{\ell}$ as the following equation
    \begin{align*}
    K_{\ell, 0} &= g_1^{\alpha^{d_{\ell}} \gamma_{ID^*_{\ell-1}}}
                   F_{1,\ell}(I_{\ell})^{-r_{\ell}}
             = g_1^{a^{d_{\ell}} \theta_{ID^*_{\ell-1}}} \prod_{j \in
                   SI^*_{ID^*_{\ell-1}}} g^{-a^{N+1-j+d_{\ell}}}
               \big( f_{1,0} \prod_{i=1}^l f_{1,i,I_1[i]} \big)^{-r_{\ell}}\\
            &= g_1^{a^{d_{\ell}} \theta_{ID^*_{\ell-1}}} \prod_{j \in
                   SI^*_{ID^*_{\ell-1}} \setminus \{ d_{\ell} \} }
                   g_1^{-a^{N+1-j+d_{\ell}}}
               \cdot g_1^{-a^{N+1}} \big( g_1^{f'_0} g_1^{a^N \Delta I_{\ell}}
               \big)^{a / \Delta I_{\ell} - r'_{\ell}} \\
            &= g_1^{a^{d_{\ell}} \theta_{ID^*_{\ell-1}}} \prod_{j \in
                   SI^*_{ID^*_{\ell-1}} \setminus \{ d_{\ell} \} }
                   g_1^{-a^{N+1-j+d_{\ell}}} (g_1^a)^{f'_0 / \Delta
                   I_{\ell}} F_{1,\ell}(I_{\ell})^{-r'_{\ell}},~\\
    K_{\ell, 1} &= g_1^{r_{\ell}}
             = (g_1^a)^{-1 / \Delta I_{\ell}} g_1^{r'_{\ell}}
    \end{align*}

Next, we show that the distribution of update keys is correct. In case of
$ID_{\ell-1} = ID_0$ and $T\neq T^*$, we have that the update key is correctly
distributed from the setting $\beta_{ID^*_{0}} = b$, $\gamma_{ID^*_0} =
\theta_{ID^*_0} - \sum_{j \in SI^*_{ID^*_0}} a^{N+1-j}$ and $r_{0} = -(-\sum_{j
\in SI^*_{ID^*_0} \setminus SI_{ID^*_0}}$ $a^{N+1-j} + \sum_{j \in SI_{ID^*_0}
\setminus SI^*_{ID^*_0}}a^{N+1-j}) / \Delta T + r'_{0}$ as the following equation
    \begin{align*}
    U_{\ell-1, 1} = &\big( g_1^{\gamma_{ID^*_0}} \prod_{j \in SI_{ID_0}}
    g_1^{\alpha^{N+1-j}}
               \big)^{\beta_{ID^*_0}} H_1(T)^{r_{1,0}}
         = \Big( g_1^{\theta_{ID^*_0}} \big( \prod_{j \in SI^*_{ID^*_0}}
                 g_1^{a^{N+1-j}}
                 \big)^{-1}
                 \prod_{j \in SI_{ID^*_0}} g_1^{a^{N+1-j}} \Big)^b
           \big( h_{1,0} \prod_{i=1}^t h_{1,i,T[i]} \big)^{r_{1,0}} \\
        = &(g_1^b)^{\theta_{ID^*_0}}
           \Big( \prod_{j \in SI^*_{ID^*_0} \setminus SI_{ID^*_0}}
                 g_1^{-a^{N+1-j}}
           \prod_{j \in SI_{ID^*_0} \setminus SI^*_{ID^*_0}} g_1^{a^{N+1-j}}
                \Big)^b \\ &\times
           \big( g_1^{h'_0} g_1^{b \Delta T} \big)^{
                -( -\sum_{j \in SI^*_{ID^*_0} \setminus SI_{ID^*_0}} a^{N+1-j} +
                    \sum_{j \in SI_{ID^*_0} \setminus SI^*_{ID^*_0}} a^{N+1-j} )
                / \Delta T + r'_{1,0} } \\
        = &(g_1^b)^{\theta_{ID^*_0}}
           \Big( \prod_{j \in SI^*_{ID^*_0} \setminus SI_{ID^*_0}} g_1^{-a^{N+1-j}}
                 \prod_{j \in SI_{ID^*_0} \setminus SI^*_{ID^*_0}} g_1^{a^{N+1-j}}
           \Big)^{-h'_0 / \Delta T} H_1(T)^{r'_{1,0}},~ \db \\
    U_{\ell-1, 2} = &g_1^{r_{1,0}}
         = \Big( \prod_{j \in SI^*_{ID^*_0} \setminus SI_{ID^*_0}} g_1^{-a^{N+1-j}}
                 \prod_{j \in SI_{ID^*_0} \setminus SI^*_{ID^*_0}} g_1^{a^{N+1-j}}
           \Big)^{-1 / \Delta T} g_1^{r'_{1,0}}.
    \end{align*}
In case of $ID_{\ell-1} \in$ \textbf{Prefix}$(ID^*_{\ell^*})$ and $T\neq
T^*$, we have that the update key is correctly distributed from the setting
$\beta_{ID_{\ell-1}} = b + \hat{\beta}_{ID_{\ell-1}}$,
$\gamma_{ID^*_{\ell-1}} = \theta_{ID^*_{\ell-1}} - \sum_{j \in
SI^*_{ID^*_{\ell-1}}} a^{N+1-j}$ and $r_{0} = -(-\sum_{j \in
SI^*_{ID^*_{\ell-1}} \setminus SI_{ID^*_{\ell-1}}}$ $a^{N+1-j} + \sum_{j \in
SI_{ID^*_{\ell-1}} \setminus SI^*_{ID^*_{\ell-1}}}a^{N+1-j}) / \Delta T +
r'_{0}$ as the following equation
    \begin{align*}
    U_{\ell-1, 1}
        =&  \big( g_1^{\gamma_{ID^*_{\ell-1}}} \prod_{j \in SI_{ID_{\ell-1}}}
                  g_1^{\alpha^{N+1-j}}\big)^{\beta_{ID^*_{\ell-1}}} H_1(T)^{r_{0}} \\
        =&  \Big( g_1^{\theta_{ID^*_{\ell-1}}} \big( \prod_{j \in SI^*_{ID^*_{\ell-1}}}
                 g_1^{a^{N+1-j}}\big)^{-1}
                 \prod_{j \in SI_{ID^*_{\ell-1}}} g_1^{a^{N+1-j}} \Big)^{b +
                    \hat{\beta}_{ID_{\ell-1}}}
           \big( h_{1,0} \prod_{i=1}^t h_{1,i,T[i]} \big)^{r_{0}} \db \\
        =&  (g_1^{b + \hat{\beta}_{ID_{\ell-1}}})^{\theta_{ID^*_{\ell-1}}}
           \Big( \prod_{j \in SI^*_{ID^*_{\ell-1}} \setminus SI_{ID^*_{\ell-1}}}
                 g_1^{-a^{N+1-j}}
           \prod_{j \in SI_{ID^*_{\ell-1}} \setminus SI^*_{ID^*_{\ell-1}}} g_1^{a^{N+1-j}}
                \Big)^{b + \hat{\beta}_{ID_{\ell-1}}} \\ &\times
           \big( g_1^{h'_0} g_1^{b \Delta T} \big)^{
                -( -\sum_{j \in SI^*_{ID^*_{\ell-1}} \setminus SI_{ID^*_{\ell-1}}}
                a^{N+1-j} +
                    \sum_{j \in SI_{ID^*_{\ell-1}} \setminus SI^*_{ID^*_{\ell-1}}}
                    a^{N+1-j} ) / \Delta T + r'_{0} } \db \\
        =&  (g_1^{b + \hat{\beta}_{ID_{\ell-1}}})^{\theta_{ID^*_{\ell-1}}}
                 \Big( \prod_{j \in SI^*_{ID^*_{\ell-1}} \setminus SI_{ID^*_{\ell-1}}}
                 g_1^{-a^{N+1-j}}
           \prod_{j \in SI_{ID^*_{\ell-1}} \setminus SI^*_{ID^*_{\ell-1}}} g_1^{a^{N+1-j}}
                \Big)^{\hat{\beta}_{ID_{\ell-1}}}\\
          &\times
           \Big( \prod_{j \in SI^*_{ID^*_{\ell-1}} \setminus SI_{ID^*_{\ell-1}}}
           g_1^{-a^{N+1-j}}
                 \prod_{j \in SI_{ID^*_{\ell-1}} \setminus SI^*_{ID^*_{\ell-1}}}
                 g_1^{a^{N+1-j}}
           \Big)^{-h'_0 / \Delta T} H_1(T)^{r'_{0}},~ \db \\
    U_{\ell-1, 2} = &g_1^{r_{0}}
         = \Big( \prod_{j \in SI^*_{ID^*_{\ell-1}} \setminus SI_{ID^*_{\ell-1}}}
         g_1^{-a^{N+1-j}}
                 \prod_{j \in SI_{ID^*_{\ell-1}} \setminus SI^*_{ID^*_{\ell-1}}}
                 g_1^{a^{N+1-j}}
           \Big)^{-1 / \Delta T} g_1^{r'_{0}},~\\
    P_0 =&  g_2^{\alpha^{N+1}\beta_{\epsilon}} \cdot \prod_{i=1}^{\ell-1}
            F_{2,i}(I_i)^{r_{i}} \cdot H(T)^{r_{0}} \cdot
            g_2^{-\alpha^{N+1} \beta_{ID^*_{\ell-1}}} \\
        =&  g_2^{a^{N+1}b} \cdot \prod_{i=1}^{\ell-1}F_{2,i}(I_i)^{r_{i}}
            \cdot H(T)^{r_{0}} \cdot g_2^{-a^{N+1}
            (b + \hat{\beta}_{ID^*_{\ell-1}})} \\
        =&  \prod_{i=1}^{\ell-1}F_{2,i}(I_i)^{r_{i}} \cdot H(T)^{r_{0}}
            \cdot g_2^{-a^{N+1} \hat{\beta}_{ID^*_{\ell-1}}}
        =   \prod_{i=1}^{\ell-1}F_{2,i}(I_i)^{r_{i}} \cdot H(T)^{r_{0}}
            \cdot e (g_1^a, g_1^{a^N})^{-\hat{\beta}_{ID^*_{\ell-1}}}.
    \end{align*}
In case of $ID_{\ell-1} = ID_0$ and $T = T^*$, we have that the update key is
correctly distributed from the setting $\beta_{ID^*_0} = b$ and $\gamma_{ID^*_0}
= \theta_{ID^*_0} - \sum_{j \in SI^*_{ID_0}} a^{N+1-j}$ as the following equation
    \begin{align*}
    U_{\ell-1, 1}  &= \big( g_1^{\gamma_{ID^*_0}} \prod_{j \in SI^*_{ID^*_0}}
                 g_1^{\alpha^{N+1-j}}
                 \big)^{\beta_{ID^*_0}} \cdot H_1(T^*)^{r_2} \\
         &= \Big( g_1^{\theta_{ID^*_0}} \big( \prod_{j \in SI^*{ID^*_0}}
                  g_1^{a^{N+1-j}}
                 \big)^{-1} \cdot \prod_{j \in SI^*_{ID^*_0}} g_1^{a^{N+1-j}}
                 \Big)^b H_1(T^*)^{r_2}
         = (g_1^b)^{\theta} H_1(T^*)^{r_2}.
    \end{align*}
In case of $ID_{\ell-1} \in$ \textbf{Prefix}$(ID^*_{\ell^*})$ and  $T = T^*$,
we have that the update key is correctly distributed from the setting
$\beta_{\epsilon} = b$, $\beta_{ID^*_{\ell-1}} = b +
\hat{\beta}_{ID^*_{\ell-1}}$ and $\gamma_{ID^* _{\ell-1}} =
\theta_{ID^*_{\ell-1}} - \sum_{j \in SI^*_{ID^*_{\ell-1}}} a^{N+1-j}$ as the
following equation
    \begin{align*}
    U_{\ell-1, 1} &= \big( g_1^{\gamma_{ID^*_{\ell-1}}} \prod_{j \in SI^*_{ID^*_{\ell-1}}}
                 g_1^{\alpha^{N+1-j}} \big)^{\beta_{ID^*_{\ell-1}}}
                 H_1(T)^{\hat{r}_0}
          = (g_1^{\theta_{ID^*_{\ell-1}}})^{b + \hat{\beta}_{ID^*_{\ell-1}}}
                 H_1(T)^{\hat{r}_0}  \\
          &= (g_1^b \cdot g_1^{\hat{\beta}_{ID^*_{\ell-1}}})^{\theta_{ID^*_{\ell-1}}}
               H_1(T^*)^{\hat{r}_0},~
    U_{\ell-1, 2}  = g_1^{\beta_{ID^*_{\ell-1}}}
         = g_1^b \cdot g_1^{\hat{\beta}_{ID^*_{\ell-1}}},~ \db \\
    P_0 &= g_2^{\alpha^{N+1}\beta_{\epsilon}} \cdot \prod_{i=1}^{\ell-1}
                   F_{2,i}(I_i)^{r_{i}} \cdot
                    H(T)^{r_{0}} \cdot g_2^{-\alpha^{N+1}\beta_{ID^*_{\ell-1}}}
             = g_2^{a^{N+1}b} \cdot \prod_{i=1}^{\ell-1}F_{2,i}(I_i)^{r_{i}}
                   \cdot H(T)^{r_{0}} \cdot g_2^{-a^{N+1}
                   (b + \hat{\beta}_{ID^*_{\ell-1}})} \\
            &= \prod_{i=1}^{\ell-1}F_{2,i}(I_i)^{r_{i}} \cdot H(T)^{r_{0}}
                    \cdot g_2^{-a^{N+1} \hat{\beta}_{ID^*_{\ell-1}}}
             = \prod_{i=1}^{\ell-1}F_{2,i}(I_i)^{r_{i}} \cdot H(T)^{r_{0}}
                   \cdot e (g_1^a, g_1^{a^N})^{-\hat{\beta}_{ID^*_{\ell-1}}}.~\\
    \end{align*}

Otherwise, the component $C$ of the challenge ciphertext is independent of
$\delta$ in the $\mc{A}$'s view since $Z_1$ is a random element in $\G_3$.
This completes our proof.
\end{proof}

\section{Conclusion}

In this paper, we first proposed an RHIBE scheme via history-preserving
updates with $O(\ell)$ number of private key elements and update key elements
by combining the BB-HIBE scheme and the BGW-PKBE scheme. Next, we proposed
another RHIBE scheme via history-free updates that reduces the number of
private key elements from $O(\ell)$ to $O(1)$.
An interesting open problem is to build an adaptive secure RHIBE scheme with
$O(\ell)$ number of private key elements and update key elements. Another one
is to construct an RHIBE scheme with $O(\ell)$ number of private key elements
and update key elements that can handle exponential number of users in the
system.

\bibliographystyle{plain}
\bibliography{rhibe-from-mlmaps}


\end{document}